\documentclass[10pt,twocolumn,journal]{IEEEtran}

\usepackage{amsmath, mathrsfs, mathtools} 
\usepackage{amsfonts}
\usepackage{amssymb}
\usepackage{color}
\usepackage{multirow}
\usepackage{stmaryrd}
\usepackage{comment}

\allowdisplaybreaks

\usepackage{caption}
\usepackage{subcaption}
\usepackage{float}

\usepackage{graphicx}
\usepackage{tcolorbox}

\newcommand{\subparagraph}{}

\usepackage[numbers,sort&compress]{natbib}

\newcommand{\Sigmay}{{\Sigmay}_{\yv}}

\usepackage[a4paper,left=1.5cm,right=1.5cm, top=1.7cm,bottom=1.7cm]{geometry}

\usepackage{multicol}

\usepackage{algorithm}
\usepackage{algpseudocode}
\usepackage{pifont}
\usepackage{varwidth}


\usepackage{mathbbol}

\def\mindex#1{\index{#1}}

\def\sq{\hbox{\rlap{$\sqcap$}$\sqcup$}}
\def\qed{\ifmmode\sq\else{\unskip\nobreak\hfil
\penalty50\hskip1em\null\nobreak\hfil\sq
\parfillskip=0pt\finalhyphendemerits=0\endgraf}\fi\medskip}

\long\def\defbox#1{\framebox[.9\hsize][c]{\parbox{.85\hsize}{%
\parindent=0pt
\baselineskip=12pt plus .1pt      
\parskip=6pt plus 1.5pt minus 1pt 
 #1}}}

\long\def\beginbox#1\endbox{\subsection*{}%
\hbox{\hspace{.05\hsize}\defbox{\medskip#1\bigskip}}%
\subsection*{}}

\def\endbox{}

\def\diag{{\text{diag}}}

\def\sign{{\rm sign}}

\newsavebox{\junk}
\savebox{\junk}[1.6mm]{\hbox{$|\!|\!|$}}

\def\argmin{\mathop{\rm arg\, min}}

\def\Re{\field{R}}

\def\bA{{\mathbb A}}

\def\bE{{\mathbb E}}

\def\sfF{{\sf F}}

\def\sfH{{\sf H}}

\def\bfmath#1{{\mathchoice{\mbox{\boldmath$#1$}}%
{\mbox{\boldmath$#1$}}%
{\mbox{\boldmath$\scriptstyle#1$}}%
{\mbox{\boldmath$\scriptscriptstyle#1$}}}}

\def\bfmY{\bfmath{Y}}

\def\bfmhhaY{\bfmath{\hhaY}} 
\def\bfmhhaY{\hbox to 0pt{$\widehat{\bfmY}$\hss}\widehat{\phantom{\raise 1.25pt\hbox{$\bfmY$}}}}

\def\til={{\widetilde =}}

 \def\FRAC#1#2#3{\genfrac{}{}{}{#1}{#2}{#3}}

\def\ddtp{{\mathchoice{\FRAC{1}{d^{\hbox to 2pt{\rm\tiny +\hss}}}{dt}}%
{\FRAC{1}{d^{\hbox to 2pt{\rm\tiny +\hss}}}{dt}}%
{\FRAC{3}{d^{\hbox to 2pt{\rm\tiny +\hss}}}{dt}}%
{\FRAC{3}{d^{\hbox to 2pt{\rm\tiny +\hss}}}{dt}}}}

\def\average#1,#2,{{1\over #2} \sum_{#1}^{#2}}

\def\eye(#1){{\bf(#1)}\quad}

\newtheorem{theorem}{{\bf Theorem}}

\newtheorem{definition}{{\bf Definition}}
\newtheorem{remark}{{\bf Remark}}

\newtheorem{lemma}{{\bf Lemma}}

\def\eq#1/{(\ref{e:#1})}

\newcommand{\beqn}[1]{\notes{#1}%
\begin{eqnarray} \elabel{#1}}

\newcommand{\eeqn}{\end{eqnarray} }

\newcommand{\beq}[1]{\notes{#1}%
\begin{equation}\elabel{#1}}

\newcommand{\eeq}{\end{equation}}

\def\bdes{\begin{description}}
\def\edes{\end{description}}

\newcounter{rmnum}

\newcounter{anum}

{\end{list}}

\def\ass(#1:#2){(#1\ref{#1:#2})}

\def\ritem#1{
\item[{\sf \ass(\current_model:#1)}]
}

\newenvironment{recall-ass}[1]{%
\begin{description}
\def\current_model{#1}}{
\end{description}
}

\usepackage{stfloats}

\usepackage{float}
\usepackage{afterpage}

\def\herm{{\sfH}}

\newcommand{\normd}[1]{{\left\vert\kern-0.25ex\left\vert\kern-0.25ex\left\vert #1 
		\right\vert\kern-0.25ex\right\vert\kern-0.25ex\right\vert}}

\renewcommand{\i}{j}

\newcommand{\E}{\mathbb{E}}
\renewcommand{\H}{\mathsf{H}}

\newcommand{\pnorm}[2]{\left\| #1 \right\|_{#2}}

\newcommand{\inner}[2]{\left\langle #1 \right\rangle_{#2}}
\newcommand{\round}[1]{\left( #1 \right)}
\renewcommand{\square}[1]{\left[ #1 \right]}

\renewcommand{\Pr}[1]{\mathrm{Pr} \square{#1}}

\newcommand{\eps}{\varepsilon}

\newcommand{\noteJ}[1]{#1}
\newcommand{\corrS}[1]{#1}
\newcommand{\id}{\mathbf{I}}

\long\def\comment#1{}

\newfont{\bb}{msbm10 scaled 1100}

\newcommand{\av}{{\bf a}}
\newcommand{\bv}{{\bf b}}

\newcommand{\hv}{{\bf h}}

\newcommand{\nv}{{\bf n}}

\newcommand{\qv}{{\bf q}}
\newcommand{\rv}{{\bf r}}
\newcommand{\sv}{{\bf s}}

\newcommand{\wv}{{\bf w}}

\newcommand{\yv}{{\bf y}}

\newcommand{\Am}{{\bf A}}
\newcommand{\Bm}{{\bf B}}
\newcommand{\Cm}{{\bf C}}

\newcommand{\Hm}{{\bf H}}

\newcommand{\Mm}{{\bf M}}

\newcommand{\Sm}{{\bf S}}

\newcommand{\Wm}{{\bf W}}

\newcommand{\Zm}{{\bf Z}}

\newcommand{\gammav}{\hbox{\boldmath$\gamma$}}

\newcommand{\tauv}{\hbox{\boldmath$\tau$}}

\newcommand{\trace}{{\hbox{tr}}}

\renewcommand{\Re}{{\rm Re}}
\renewcommand{\Im}{{\rm Im}}

\newcommand{\transp}{{\sf T}}
\renewcommand{\vec}{{\rm vec}}

\newcommand{\plotwidth}{0.8}

\makeatletter
\def\ps@IEEEtitlepagestyle{
  \def\@oddfoot{\mycopyrightnotice}
  \def\@evenfoot{}
}
\def\mycopyrightnotice{
  {\footnotesize This work has been submitted to the IEEE for possible publication. Copyright may be transferred without notice, after which this version may no longer be accessible.\hfill} 
  \gdef\mycopyrightnotice{}
}

\title {Plug-in Channel Estimation with Dithered Quantized Signals in Spatially Non-Stationary Massive MIMO Systems}
\author{Tianyu Yang$^1$, Johannes Maly$^{2,3}$, Sjoerd Dirksen$^4$,  \\and Giuseppe Caire$^1$  
\thanks{$^1$Communications and Information Theory Group (CommIT), Technische Universit\"{a}t Berlin, 10587 Berlin, Germany (e-mail: \{tianyu.yang,  caire\}@tu-berlin.de).}
\thanks{$^2$Ludwig-Maximilians-Universität Munich, 80333 Munich, Germany  (e-mail: maly@math.lmu.de).}
\thanks{$^3$Munich Center for Machine Learning (MCML).}
\thanks{$^4$Utrecht University, 3584 CD Utrecht, Netherlands (e-mail: s.dirksen@uu.nl).}
}

\begin{document}

\maketitle	


\begin{abstract}
As the array dimension of massive MIMO systems increases to unprecedented levels, two problems occur. First, the spatial stationarity assumption along the antenna elements is no longer valid. Second, the large array size results in an unacceptably high power consumption if high-resolution analog-to-digital converters are used. To address these two challenges, we consider a Bussgang linear minimum mean square error (BLMMSE)-based channel estimator for large scale massive MIMO systems with one-bit quantizers and a spatially non-stationary channel. Whereas other works usually assume that the channel covariance is known at the base station, we consider a plug-in BLMMSE estimator that uses an estimate of the channel covariance and rigorously analyze the distortion produced by using an estimated, rather than the true, covariance.
To cope with the spatial non-stationarity, we introduce dithering into the quantized signals and provide a theoretical error analysis. In addition, we propose an angular domain fitting procedure which is based on solving an instance of non-negative least squares. For the multi-user data transmission phase, we further propose a BLMMSE-based receiver to handle one-bit quantized data signals. Our numerical results show that the performance of the proposed BLMMSE channel estimator is very close to the oracle-aided scheme with ideal knowledge of the channel covariance matrix. The BLMMSE receiver outperforms the conventional maximum-ratio-combining and zero-forcing receivers in terms of the resulting ergodic sum rate.
\end{abstract}


\section{Introduction}
Massive multiple-input-multiple-output (MIMO) has been vastly researched and considered as an essential technology in 5G wireless communication systems within sub-6 GHz bands \cite{lu2014overview,bjornson2019massive,li201712}. Benefiting from the large number (tens to hundreds) of antennas at the base station (BS) array, dozens of users can be served in the same time-frequency slots. This results in higher spectrum and energy efficiency due to the spatial multiplexing and high array gain \cite{lu2014overview, ngo2013energy}. Theory shows that by increasing the array dimension, i.e., the number of antenna elements, it is possible to achieve higher data rates and to mitigate the impacts of inter-cell interference and thermal noise \cite{marzetta2010noncooperative}.  Nevertheless, as the array dimension increases, two new challenges occur. First, some inherent properties of the channel environment change compared to small-scale MIMO, so that the basic assumptions of massive MIMO design are no longer valid for large arrays. Specifically, most existing massive MIMO works are based on the assumption of a spatially stationary channel, where all antenna elements observe the same far-field propagation from the channel scatters \cite{haghighatshoar2016massive, haghighatshoar2017massive, haghighatshoar2018low, khalilsarai2020structured}. However, in the very large antenna array regime, spatial non-stationarity has been experimentally observed  \cite{payami2012channel}. Two main reasons for this non-stationarity are further discussed in \cite{gao2013massive, gao2015massive, wu2014non}. First, with large arrays, some scattering clusters may be located in the Fresnel region, where the distance between them and the BS array is smaller than the Fraunhofer distance (but still larger than the Fresnel distance) \cite{selvan2017fraunhofer}. For instance, for a massive MIMO system with a 3GHz carrier frequency and 256 antennas in a uniform linear array (ULA), the resulting Fraunhofer distance is  $d_{\rm F} = \frac{2D^2}{\lambda} \approx $ 3km, where $D$ is the total length of the antenna array and $\lambda$ is the wavelength. As a consequence, user signals impinging onto the BS array cannot be assumed to follow far-field propagation and have a spherical wavefront instead of a plane wavefront. In such cases, the receivers are in the radiative near-field (corresponding to the Fresnel region) and recent research on near-field communication can be found in e.g., \cite{bjornson2021primer,bjornson2020power}.   Second, due to the physically large size of the array, some of the scattering clusters may only be visible to a part of the array. Furthermore, a new deployment of large arrays that are usually integrated into large structures, e.g., along the walls of buildings \cite{medbo2014channel,bjornson2019massivenext}, was considered as an extension of massive MIMO and referred to as extra-large scale massive MIMO (XL-MIMO) in \cite{de2020non}. It is also pointed out in \cite{de2020non} that due to the large dimension (tens of meters) of XL-MIMO, spatially non-wide sense stationary (non-WSS) characteristics appear along the array. 

Aside from the spatially non-WSS property of large antenna arrays, a second major concern is the hardware cost and power consumption of high-resolution analog-to-digital converters (ADCs). Commercial high-resolution ADCs (12 to 16 bits) are expensive and their power consumption grows exponentially in terms of the number of quantization bits \cite{walden1999analog}. This problem is even more severe for wideband systems, where the power consumption of high-resolution ADCs increases linearly with the signal bandwidth due to required higher sampling rates \cite{murmann2015race}. To alleviate the issue of high power consumption, low-resolution ADCs (e.g., 1-3 bits) are utilized for massive MIMO systems \cite{mezghani2008analysis, nossek2006capacity,singh2009limits} and it is shown in \cite{mezghani2008analysis,nossek2006capacity} that the capacity loss due to the coarse quantization is approximately equal to only $\pi/2$ at low signal-to-noise ratios (SNRs). In massive MIMO systems the SNR per antenna element may be relatively low, while still achieving an overall large spectral efficiency over data stream due to the large number of antennas per user (data stream), such that both spatial multiplexing gain and array gain are achieved. 

\subsection{Related work}

Accurate estimation of channel state information (CSI) at the BS is a key factor to achieve the potential benefits of massive MIMO systems. Taking the spatially non-WSS property into account, an adaptive grouping sparse Bayesian learning scheme was proposed for uplink channel estimation in \cite{cheng2019adaptive}. A model-driven deep learning-based channel reconstruction scheme was proposed in \cite{han2020deep}. On the other hand, many recent works have investigated channel estimators with one-bit quantized signals in massive MIMO systems, see e.g., \cite{li2017channel,wan2020generalized} and references therein.  To the best of our knowledge, no work has addressed the problem of channel estimation with both low-resolution quantization and spatially non-WSS channels in massive MIMO systems. To fill this gap, in this paper we adopt the Bussgang linear minimum mean square estimation (BLMMSE) method that was initially proposed in \cite{li2017channel}, and propose a BLMMSE-based ``plug-in'' channel estimator for one-bit Massive MIMO systems with the spatially non-WSS property.\footnote{By ``plug-in'' we mean that the BLMMSE requires the knowledge of the channel covariance, which is typically assumed to be known. However, in our setting, also the channel covariance needs to be estimated from low-resolution quantized samples. The plug-in estimator consists of using the estimated covariance ``as if'' it were the true one, in the BLMMSE.}

Since our proposed scheme involves channel covariance estimation from one-bit quantized samples, we briefly introduce the existing works that are based on the modification of the classical arcsine law. The first work line proposes to add \textit{biased Gaussian dithering} to treat stationary \cite{eamaz2021modified,eamaz2022covariance} and non-stationary \cite{eamaz2023covariance} Gaussian processes. Due to the Gaussian nature of the dithering thresholds, the authors can derive a modified arcsine law that relates the covariance matrix of the quantized samples to an involved integral expression depending on the entries of the unquantized covariance matrix. Different numerical methods are further proposed for evaluating/inverting the integral to finally recover the covariance matrix. 

Following the observation that variances can be estimated from one-bit samples with fixed dithering threshold \cite{chapeau2008fisher,fang2010adaptive}, assuming an underlying stationary Gaussian process the authors in another line of work \cite{liu2021one} explicitly calculate the mean and autocorrelation of one-bit samples with \textit{a fixed dither} and further design an algorithm for estimating the covariance matrix entries of the underlying process. The variance of maximum-likelihood estimators of the quantities computed in \cite{liu2021one} is further analyzed in \cite{xiao2023one} using low-order Taylor expansions, where the authors conclude that a fixed dither yields suboptimal results if the variances of the process are strongly fluctuating. They then propose to use a dither that \textit{increases monotonically in time}.

The aforementioned works have two main drawbacks. First, their theoretical analysis uses tools that are specific to Gaussian samples and hence their results are tailored to Gaussian samples. Second, their results concern the large sample limit: they do not analyze the approximation error in terms of the number of samples, i.e., they leave open which error is caused by using the sample covariance matrix instead of the true covariance matrix of the quantized samples. In contrast to the aforementioned works, the results in \cite{dirksen2021covariance}, which we extend in this paper, provide the first non-asymptotic (and near-optimal) guarantees in the literature. Furthermore, the use of uniform random dithering vectors (instead of fixed, deterministically varying, or random Gaussian ones) considerably simplifies the analysis and allows to obtain results for non-Gaussian samples.

\subsection{Contributions}
 Our main contributions are summarized as follows.
\begin{itemize}
    \item We adopt a BLMMSE channel estimator to deal with the one-bit quantized signal. However, in contrast to \cite{li2017channel} that assumes exact knowledge of the channel covariance at the BS, we propose a ``plug-in'' version that instead uses an estimate of the channel covariance. Our first contribution is a theoretical analysis of the distortion caused by the use of this estimate: we estimate the mean squared distance between the BLMMSE estimate based on an estimate of the channel covariance versus the BLMMSE estimate based on the true channel covariance (see Lemma~\ref{lem:Stability}). 
    
    \item Our second contribution is a method to estimate the channel covariance matrix based on one-bit samples. We introduce dithering into the one-bit ADCs to cope with the non-Toeplitz structure of the channel covariance matrix (resulting from the spatially non-WSS channel). We propose a covariance estimator based on dithered quantized samples and derive bounds on the estimation error in terms of the maximum norm and the Frobenius norm (Theorem~\ref{thm:FrobeniusDithered}). By combining this result with the aforementioned bound on the MSE achieved by the BLMMSE with given estimated covariance, we derive a bound on the expected MSE of the channel estimator in terms of the number of samples used to estimate the channel covariance (Theorem~\ref{thm:MainResult}). 
    
    \item We empirically further enhance the proposed channel covariance estimator by exploiting the angle domain of the spatially non-WSS channel. Using dictionary functions in the angle domain, we formulate the channel covariance estimation as a non-negative least-squares problem (NNLS), which can be efficiently solved by a standard numerical NNLS solver, e.g., \cite{bill2022nnls}, even for very large problem dimensions.

    \item We design a linear receiver for the uplink (UL) data transmission phase, based on an estimate of the channel matrix obtained in the training phase, to achieve better rate detection and thus improve the ergodic sum rate of multi-user severing. Contrary to the conventional maximum-ratio-combining (MRC) and zero-forcing (ZF) receivers that do not take the quantization into account, the proposed receiver considers the Bussgang decomposition of one-bit quantized data signals and uses BLMMSE-based estimation with knowledge of only the estimated channel covariance matrix.

\end{itemize}
Our numerical results show that the proposed BLMMSE channel estimator, which uses the proposed channel covariance estimator based on dithered quantized samples, is superior to benchmark methods and achieves a performance very close to the performance of an oracle-aided scheme using the true channel covariance matrix. The proposed BLMMSE-based receiver also significantly outperforms MRC and ZF receivers as expected due to its specific consideration of quantized signals.

\subsection{Organization}
The rest of this paper is organized as follows. In Section~\ref{sec:nonWSS}, we introduce the channel model with the spatially non-stationary property. Section~\ref{sec:channelEst} is devoted to the analysis of the BLMMSE channel estimator and our results on channel covariance estimation from one-bit quantized samples. In Section~\ref{sec:receiver}, we propose the BLMMSE receiver for the data transmission phase to obtain a higher sum rate. The numerical results are then provided in Section~\ref{sec:numerics}. Finally, in Section~\ref{sec:discussion} we conclude our work and provide a discussion of possible future research directions. 

\subsection{Notation}

For any $N\in \mathbb{N}$ we write $[N]=\{1, 2, \dots, N\}$. We use lower-case, bold lower-case, and bold upper-case letters to denote scalars, column vectors, and matrices, respectively.  The trace, transpose and Hermitian transpose are respectively denoted by $\trace(\cdot)$, $(\cdot)^\transp$ and $(\cdot)^\herm$. $\bE[\cdot]$ returns the mathematical expectation. $\diag(\Am)$ gives a diagonal matrix with diagonal of $\Am$, while  $\diag(\av)$ denotes the diagonal matrix with diagonal equal to $\av$. We denote the $M\times M$ identity matrix by $\mathbf{I}_M$. The $i$-th element of a vector $\av$ is denoted by $[\av]_i$, while the $i$-th row and column of a matrix $\Am$ are respectively denoted by $[\Am]_{i, \cdot}$ and $[\Am]_{\cdot, i}$. An all-zero matrix is denoted by $\mathbf{0}$. $\|\av\|_2$ denotes the Euclidean norm of a vector $\av$. $\|\Am\|_{\sf F}$, $\|\Am\|_{}$, and $\|\Am\|_{\infty}$ denote the Frobenius, operator, and maximum norms of a matrix $\Am$. 
We use $\langle \Am,\Bm \rangle_{\sf F} := \trace(\Am^\herm \Bm)$ to denote the Frobenius  inner product. We furthermore use $a \lesssim b$ to abbreviate $a \le Cb$, for some absolute constant $C > 0$. 

\section{System model with Spatially Non-Stationary Channel}
\label{sec:nonWSS}

\begin{figure}
	\begin{center}
	\includegraphics[width=\plotwidth\columnwidth]{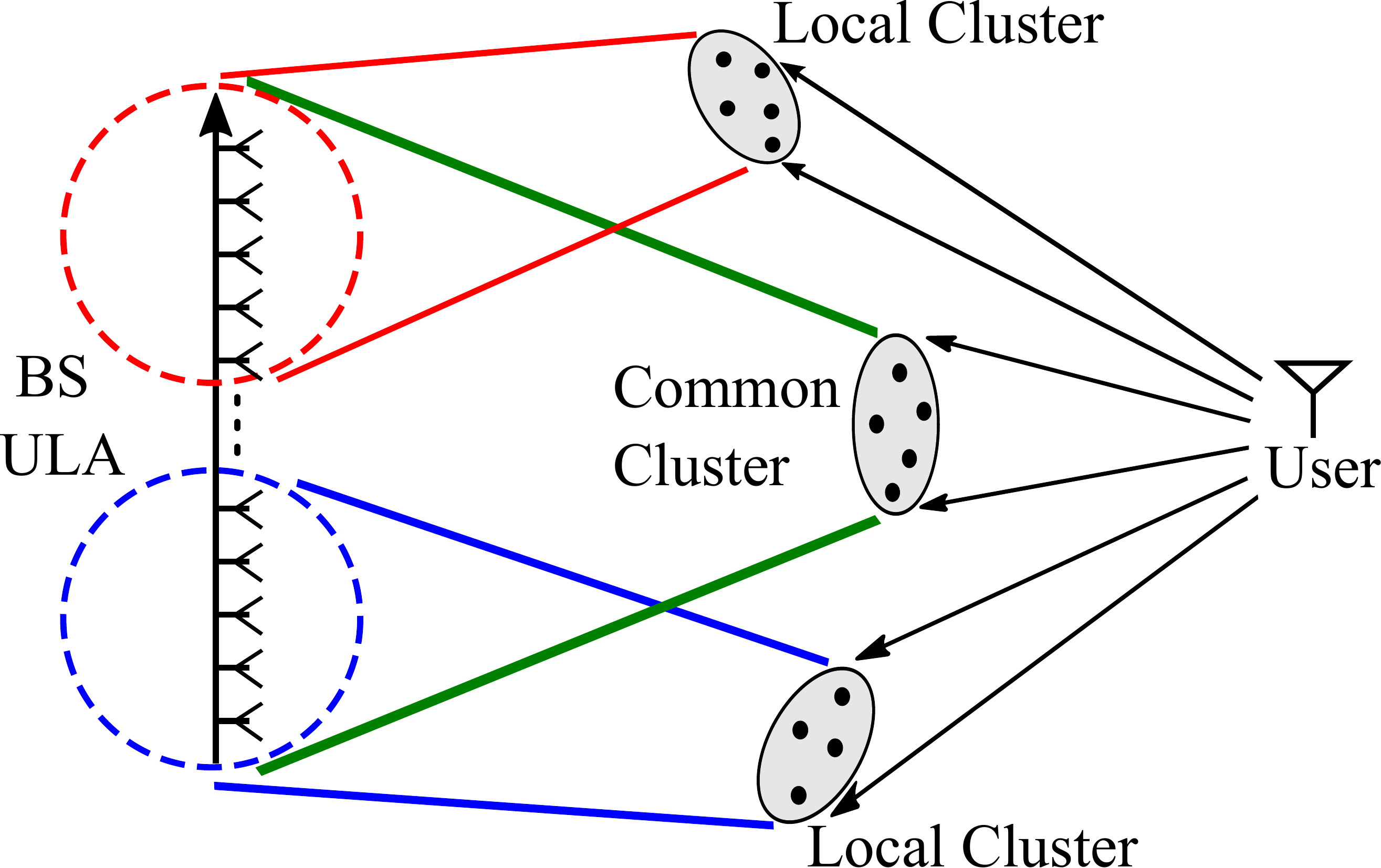}
	\caption{Illustration of the studied large-scale Massive MIMO system in ULA with the spatially non-WSS channel, where local clusters are only visible to a part of antenna elements while the common clusters are visible to the whole array.}\label{fig:cluster}
	\end{center}
 \vspace{-.5cm}
\end{figure} 
We consider a single-cell massive MIMO communication system, where a BS equipped with $M$ antennas in a uniform linear array (ULA) communicates with multiple users through a multipath channel. We assume that the channel scattering clusters consist of common and local clusters, where common clusters are visible to all antennas while local clusters are only visible to a sub-array. An illustration of the considered scattering geometry is shown in Fig.~\ref{fig:cluster}. The channel vector resulting from the contribution of the common clusters at the $n$-th time slot is given by
\begin{equation}
    \hv^\text{c}_n = \sum_{i=1}^{L^{\text{c}}}\rho^\text{c}_i(n) \av(\theta^{\text{c}}_i),
\end{equation}
where $L^{\text{c}}$ denotes the total number of multipaths in the common clusters,  $\rho^\text{c}_i(n) \sim \mathcal{CN}(0, \gamma^{\text{c}}_i)$ is the $i$-th complex channel gain of the common clusters with its power $\gamma^{\text{c}}_i$, $\theta^{\text{c}}_i$ is the $i$-th angle of arrival (AoA), and where $\av(\theta)\in\mathbb{C}^{M\times1}$ is the steering vector, whose $m$-th entry is $[\av(\theta)]_m~=~e^{j\pi (m-1) \sin(\theta)}$ by assuming that the antenna spacing is equal to half of the carrier wavelength, for all $m \in \mathcal{M}$, where $\mathcal{M}=[M]$ is the antenna index set of all $M$ antennas. 

Assume that there are $L$ local clusters and that each local cluster is visible to a consecutive sub-array. The $i$-th local cluster is thus visible to $M^\text{l}_i$ antennas with index set $\mathcal{M}^\text{l}_i$, where $M^\text{l}_i = |\mathcal{M}^\text{l}_i|$. The channel vector resulting from the contribution of the paths in the $i$-th local cluster at the $n$-th time slot is given by 
\begin{equation}
    \hv^\text{l}_{n,i} = \sum_{j=1}^{L^{\text{l}}_i}\rho^\text{l}_{ij}(n) \Sm_i\av(\theta^{\text{l}}_{ij}), \quad i \in [L],
\end{equation}
where $L^{\text{l}}_i$ denotes the total number of multipaths in the $i$-th local cluster,  $\rho^\text{l}_{ij}(n) \sim \mathcal{CN}(0, \gamma^{\text{l}}_{ij})$ is the $j$-th complex channel gain of the $i$-th local cluster with its power $\gamma^{\text{l}}_{ij}$, $\theta^{\text{l}}_{ij}$ is the AoA, and where $\Sm_i\in\mathbb{C}^{M\times M}$ is the diagonal selection matrix indicating the visible sub-array of the $i$-th local cluster, whose diagonal is defined as
\begin{equation}
    [\Sm_i]_{m,m} = 
    \begin{cases} 
       1, &  \; m \in \mathcal{M}^\text{l}_i \\
      0, & \; m\in\mathcal{M}\setminus\mathcal{M}^\text{l}_i.
   \end{cases}
\end{equation}
We further assume that the channel gains of different paths in all common and local clusters at each time slot $n$,  $\left\{\rho^{\text{c}}_i(n)\right\}_{i=1}^{L^{\text{c}}}$ and $\left\{\rho^{\text{l}}_{ij}(n)\right\}_{j=1}^{L^{\text{l}}_i}$, $\forall i \in [L]$, are uncorrelated\footnote{Note that this is a standard assumption specified in, e.g., the channel models of 3GPP standard TR 38.901 \cite{3gpp38901} and TR 25.996 \cite{3gpp25996}. This assumption is also implicitly included in the documentation of the well-known channel simulator QuaDRiGa \cite{jaeckel2014quadriga}.}.  Note that we implicitly assume that the channel geometry and visibility of all clusters do not change over the channel geometry coherence time $T_c$, which is a much longer time period than the channel coherence time (see \cite{va2016impact} and the references therein). Concretely, the small-scale channel coefficients $\{\rho^{\rm c}_i(n)\}_{i=1}^{L^{\rm c}}$ and $ \{\rho^{\rm l}_{ij}(n)\}_{j=1}^{L^{\rm l}_i}$, $i \in [L]$, are stationary over $T_c$, and  the Angular Power Spectrum (APS) $\left\{\gamma^{\text{c}}_i\right\}_{i=1}^{L^{\text{c}}}$ and $\left\{\gamma^{\text{l}}_{ij}\right\}_{j=1}^{L^{\text{l}}_i}$, $i \in [L]$, along with the AoAs $\left\{\theta^{\text{c}}_i\right\}_{i=1}^{L^{\text{c}}}$ and $\left\{\theta^{\text{l}}_{ij}\right\}_{j=1}^{L^{\text{l}}_i}$, $i \in [L]$, as well as the selection matrices $\{\Sm_i\}_{i=1}^{L}$ are constant over $T_c$.
Under these assumptions, the total channel vector at $n$-th time slot $\hv_n$ and the corresponding total channel covariance matrix $\Cm_\hv$ are given by
\begin{align}
    \hv_n &= \hv^\text{c}_n + \sum^L_{i=1} \hv^{\text{l}}_{n,i},\\
    \Cm_\hv &= \mathbb{E}[\hv_n\hv_n^\herm] = \Cm_{\hv^{\text{c}}} + \sum^L_{i=1} \Cm_{\hv^{\text{l}}_i} \label{eq:channelCov_true} \\
    &=\sum_{i = 1}^{L^{\text{c}}}\gamma^\text{c}_i\av(\theta^{\text{c}}_i)\av(\theta^{\text{c}}_i)^\herm +\sum^L_{i=1}\sum_{j=1}^{L^{\text{l}}_i}\gamma^\text{l}_{ij}\Sm_i\av(\theta^{\text{l}}_{ij})\av(\theta^{\text{l}}_{ij})^\herm\Sm_i^\herm \nonumber \\
    &= \Am^{\text{c}} \diag(\gammav^{\text{c}})(\Am^{\text{c}})^\herm + \sum^L_{i=1}\Sm_i\Am^{\text{l}}_i \diag(\gammav^{\text{l}}_i)(\Am^{\text{l}}_i)^\herm\Sm_i^\herm,  \nonumber
\end{align}
where we define $\Am^{\text{c}} := [\av(\theta^{\text{c}}_1), \dots, \av(\theta^{\text{c}}_{L^{\text{c}}})], \gammav^{\text{c}} := [\gamma^{\text{c}}_1,\dots,\gamma^{\text{c}}_{L^{\text{c}}}]^\transp$ for common clusters, and $\Am^{\text{l}}_i := [\av(\theta^{\text{l}}_{i,1}), \dots, \av(\theta^{\text{l}}_{i,L^{\text{l}}_i})],\gammav^{\text{l}}_i := [\gamma^{\text{l}}_{i,1},\dots,\gamma^{\text{l}}_{i, L^{\text{l}}_i}]^\transp, \;\forall i \in [L]$ for local clusters.

The total channel power gain of all common clusters and all local clusters are given by  
\begin{align}
    P^\text{c} = \sum_{i=1}^{L^{\text{c}}}\gamma^\text{c}_i
    \qquad \text{and} \qquad
    P^\text{l} = \sum^L_{i=1}P^\text{l}_i = \sum^L_{i=1}\sum_{j=1}^{L^{\text{l}}_i}\gamma^\text{l}_{ij},
\end{align}
where we assume that $P^\text{c}$ and $P^\text{l}$ are normalized such that $\max(\diag(\Cm_{\hv})) = 1$. To help readability, Table~\ref{tab:notation} summarizes the model notation.
\begin{table}[!ht]
    \centering
    \begin{tabular}{|c|c|}
    \hline
    $L$  & Number of local clusters \\
         $L^{\text{c}}$, $L^{\text{l}}_i$  & Number of multipaths of common and local clusters \\
         $\rho^{\text{c}}_i$, $\rho^{\text{l}}_{ij}$ & Complex channel gain of common and local clusters \\
         $\gamma^{\text{c}}_i$, $\gamma^{\text{l}}_{ij}$ & APS of common and local clusters \\
         $\Sm_i$ & Diagonal selection matrices of local clusters \\
         $\theta^{\text{c}}_i$, $\theta^{\text{l}}_{ij}$ & AoAs of common and local clusters \\
         $\Am^{\text{c}}$, $\Am^{\text{l}}_{i}$ & Matrices of steering vectors of common and local clusters \\
    \hline
    \end{tabular}
    \caption{Summary of the used notations}
    \label{tab:notation}
    \vspace{-.5cm}
\end{table}

\section{Channel estimation with one-bit samples}
\label{sec:channelEst} 
We follow the common assumption that the pilot sequences of different users are orthogonal to each other in the time-frequency domain. Since the observations used for channel estimation use only the UL pilots, we focus on a generic user under a normalized pilot. Then, at the $n$-th time-frequency resource block (RB), the BS receives the signal
\begin{equation}\label{eq:received_signal}
    \yv_n = \hv_n + \nv_n,
\end{equation}
where $\hv_n \sim\mathcal{CN}(\mathbf{0},\Cm_{\hv})$ is the $M\times1$ channel vector and $\nv_n \sim\mathcal{CN}(\mathbf{0},N_0\mathbf{I}_M)$ is additive white Gaussian noise (AWGN) with noise power $N_0$. The SNR is thus defined by $1/N_0$ due to the assumption that $\max(\diag(\Cm_{\hv})) = 1$.
After one-bit ADC the quantized signal  becomes
\begin{equation}
    \rv_n = \mathcal{Q}(\yv_n),
\end{equation}
where $\mathcal{Q}(\cdot)$ is a suitable one-bit scalar quantizer that is applied elementwise and separately to the real and imaginary parts. One popular instance for such a quantizer is the complex-sign operator $\operatorname{csign}$ \cite{bar2002doa,risi2014massive}, given as 
\begin{align}
    \rv^{\mathrm{nd}}_n  =\operatorname{csign}(\yv_n)= \frac{1}{\sqrt{2}}\big(\sign(\Re(\yv_n)) + j \;\sign(\Im(\yv_n))\big ),
    \label{eq:r-nd}
\end{align}
which quantizes the entries of $\Re(\yv_n)$ and $\Im(\yv_n)$ independently, i.e., the $\sign$-function
\begin{align}
    \sign \colon \mathbb R \to \{-1,1\}
    \qquad 
    \sign(x) = \begin{cases}
       1 & x \ge 0 \\
       -1 & x < 0
    \end{cases}
\end{align}
acts componentwise (\emph{memoryless scalar quantization}). We use the superscript ``nd'' (``non-dithered'') in \eqref{eq:r-nd} since the $\sign$-function is applied directly to the samples without dithering. Note that this type of one-bit quantization looses any scaling information.

\subsection{Bussgang LMMSE channel estimator} 
We consider channel estimation for a generic time slot. Thus, we ignore the time slot index $n$ for simplicity. In order to estimate the channel vector $\hv$ from a quantized sample $\rv$, we first transfer the nonlinear quantizer operation to a statistically equivalent linear formulation via the well-known \emph{Bussgang decomposition} \cite{bussgang1952crosscorrelation}, which yields
\begin{align}
    \rv &= \mathcal{Q}(\yv) = \Am \yv + \qv = \Am \hv + \Am \nv + \qv = \Am \hv + \widetilde{\nv}\label{eq:bussgangchannel},
\end{align}
where the linear operator $\Am$ is called the \emph{Bussgang gain}, $\qv$ is a mean-zero random vector that is uncorrelated with $\yv$, and $\widetilde{\nv} := \Am \nv + \qv$ is the total noise. To enforce $\qv$ to be uncorrelated with $\yv$, the Bussgang gain $\Am$ is chosen to minimize the power of the equivalent quantization noise \cite{demir2020bussgang} such that
\begin{align}\label{eq:Am}
    \Am = \mathbb{E}\left[\rv\yv^\herm\right]\mathbb{E}\left[\yv\yv^\herm\right]^{-1}= \Cm_{\rv\yv}\Cm^{-1}_{\yv},
\end{align}
where $\Cm_{\rv\yv}=\mathbb{E}\left[\rv\yv^\herm\right]$ denotes the covariance between the quantized signal $\rv$ and the  received signal $\yv$. The so-called BLMMSE estimator \cite{li2017channel} of the channel vector $\hv$ given the quantized signal $\rv$ is then expressed as
\begin{equation}\label{eq:mmse}
    \widehat{\hv}^{\text{BLM}} = \Cm_{\hv\rv} \Cm_{\rv}^{-1} \rv.
\end{equation}
Note that this is not the optimal MMSE estimator since $\qv$ is not Gaussian noise. We however know that the vector $\qv$ is uncorrelated with the vector $\yv$ and one can prove that $\qv$ is also uncorrelated with the channel vector $\hv$, see \cite[App.\ A]{li2017channel} for the proof of $\mathbb{E} [\hv \qv^\herm] = \mathbf{0}$. Thus, $\hv$ is uncorrelated with the total noise $\widetilde{\nv}$ and consequently we obtain from \eqref{eq:bussgangchannel} that
\begin{equation}
    \Cm_{\hv\rv} =  \Cm_{\hv} \Am^\herm.
\end{equation}
Similarly as in \cite{li2017channel}, $\Am$ and $\Cm_{\rv}$ can be easily computed as follows. For the one-bit quantizer in \eqref{eq:r-nd} and Gaussian inputs,  $\Cm_{\rv\yv}$ is given as \cite{mezghani2012capacity},  \cite[Ch.12]{papoulis2002probability}
\begin{align}
\label{eq:cyr}
    \Cm_{\rv\yv} = \sqrt{\frac{2}{\pi}} \diag(\Cm_{\yv})^{-\frac{1}{2}}\Cm_{\yv}
\end{align}
and combining \eqref{eq:cyr} and \eqref{eq:Am} we obtain 
\begin{align}
\label{eq:Am_fianl}
    \Am = \sqrt{\frac{2}{\pi}} \diag(\Cm_{\yv})^{-\frac{1}{2}}.
\end{align}
Furthermore, we have to calculate the covariance matrix $\Cm_{\rv}$. We adopt the well-known ``arcsine law'' to obtain $\Cm_{\rv}$ from $\Cm_{\yv}$. The arcsine law reveals an exact nonlinear invertible connection between the normalized autocorrelation functions of the one-bit quantized signal and the original unquantized signal. It was initially introduced in \cite{van1966spectrum} for real signals and then extended in \cite{jacovitti1994estimation} for complex signals. Specifically in our case, $\Cm_{\rv}$ can be obtained using the map $\mathcal{P}_{\text{arcsine}}(\cdot)$ by the arcsine law  as \cite{bar2002doa, mezghani2012capacity} 
\begin{align}
\label{eq:arcsin_law}
    \Cm_{\rv} &= \mathcal{P}_{\text{arcsine}}(\Cm_{\yv}) \\
    &=\frac{2}{\pi} \Bigg(\mathrm{arcsin}\left(\diag(\Cm_{\yv})^{-\frac{1}{2}}\Re(\Cm_{\yv})\diag(\Cm_{\yv})^{-\frac{1}{2}}\right) \nonumber\\
    &\quad \quad 
    + j \; \mathrm{arcsin}\left(\diag(\Cm_{\yv})^{-\frac{1}{2}}\Im(\Cm_{\yv})\diag(\Cm_{\yv})^{-\frac{1}{2}}\right)\Bigg). \nonumber
\end{align}

With the BLMMSE estimator in hand, \eqref{eq:mmse} has a closed form that only depends on $\Cm_{\yv} = \Cm_{\hv} + N_0\mathbf{I}_M$. Whereas the noise power $N_0$ is normally assumed to be known at the BS\footnote{This can be achieved via, e.g., low rate control channel.}, the channel covariance matrix $\Cm_{\hv}$ still needs to be estimated from samples to finally apply the BLMMSE estimator \eqref{eq:mmse}. Given any ``plug-in estimator’’ $\widehat{\Cm}_{\yv}$ of $\Cm_{\yv}$, we define the associated estimated channel vector as 
\begin{equation}
\label{eq:mmseEstCov}
    \widehat{\hv} = \widehat{\Cm}_{\hv\rv} \widehat{\Cm}_{\rv}^{-1} \rv,
\end{equation}
where $\widehat{\Cm}_{\hv\rv}$ and $\widehat{\Cm}_{\rv}$ are the estimators of $\Cm_{\hv\rv}$ and $\Cm_{\rv}$ obtained by replacing $\Cm_{\yv}$ by its estimator $\widehat{\Cm}_{\yv}$, i.e., $\widehat{\Cm}_{\hv\rv} =  \widehat{\Cm}_{\hv} \widehat{\Am}^\herm, \widehat{\Cm}_{\hv} = \widehat{\Cm}_{\yv} - N_0\mathbf{I}_M, \widehat{\Am} = \sqrt{\frac{2}{\pi}} \diag(\widehat{\Cm}_{\yv})^{-\frac{1}{2}},$ and $\widehat{\Cm}_{\rv} = \mathcal{P}_{\text{arcsine}}(\widehat{\Cm}_{\yv})$.
The following lemma controls the estimation error in \eqref{eq:mmse} if an estimator of $\Cm_{\yv}$ is used. Note that this result is generic, in the sense that it is valid for \emph{any} estimator $\widehat{\Cm}_{\yv}$ of $\Cm_{\yv}$. In the next section, we will apply this result for a specific estimator based on one-bit quantized samples.

\begin{lemma} \label{lem:Stability}
  There are absolute constants $c_1,c_2,C>0$ such that the following holds. Let $\theta \in (0,1)$ be fixed. Assume that
\begin{align}
\label{eqn:nondiagBoundLemma}
    &\left| \left[\diag(\Cm_{\yv})^{-\frac{1}{2}} \Cm_{\yv} \diag(\Cm_{\yv})^{-\frac{1}{2}}\right]_{i,j} \right| \le 1-\theta,
    \quad \forall i \neq j \\
&\text{and}\qquad\min_{i\in [M]} |[\Cm_{\yv}]_{i,i}| \ge \theta, \qquad \lambda_{\min}(\Cm_{\rv}) \ge \theta, \label{eqn:lowBoundLemma}
\end{align} 
where $\lambda_{\min}(\cdot)$ gives the minimal eigenvalue of the matrix.
Consider $\eps_{\sf F}>0,\eps_{\infty}>0$ such that 
\begin{equation}
\| \widehat{\Cm}_{\yv} - \Cm_{\yv} \|_{\sf F}<\eps_{\sf F}, \qquad \| \widehat{\Cm}_{\yv} - \Cm_{\yv} \|_{\infty}<\eps_{\infty}
\end{equation}
and assume that 
\begin{equation}\label{eqn:epsinftyepsF}
\begin{aligned}
&\eps_{\infty}\leq c_1\min\left\{\frac{\eps_{\sf F}}{\|\Cm_{\yv}\|_{\sf F}},\frac{\theta^3}{\|\Cm_{\yv}\|_{\infty}}, \theta\right\} \\ \text{and} \ \ 
&\eps_{\sf F}\leq  \corrS{c_2 \min \left\{\theta^4, \frac{\theta^6\|\Cm_{\hv}\|_{\sf F}}{\max\{1,\|\Cm_{\hv}\|\} \ \|\Cm_{\yv}\|} \right\}}.
\end{aligned}
\end{equation}
Then,  
    \begin{align}
        \E \left[\left\| \widehat\hv - \widehat{\hv}^{\text{BLM}} \right\|_{2}^2\right] \le C \theta^{-6} \noteJ{\max\{1,\|\Cm_{\hv}\|_{}\} \ \|\Cm_{\hv}\|_{\sf F}} \eps_{\sf F},
    \end{align}
where the expectation is taken with respect to $\rv$.
\end{lemma}
\begin{proof}
     See Appendix \ref{sec:AppendixStability}.    
\end{proof}
\begin{remark}
Let us briefly comment on the assumptions in Lemma~\ref{lem:Stability}. As the construction of the BLMMSE involves the inverses of $\diag(\Cm_{\yv})$ and $\Cm_{\rv}$, it is to be expected that in the situation that these matrices are near-singular, a small error in the estimation of the covariance can lead to a large difference between $\widehat\hv$ and $\widehat{\hv}^{\text{BLM}}$. This expected behaviour is quantified in Lemma~\ref{lem:Stability} using the parameter $\theta$. The lower bound on $\lambda_{\min}(\Cm_{\rv})$ is an implicit condition on $\Cm_{\yv}$. To give a more explicit condition, let us write $\operatorname{offdiag}(\Cm_{\rv})$ for the off-diagonal part. Using that $\|\arcsin(\Bm)\|\leq \frac{\pi}{2}\|\Bm\|$ if $\|\Bm\|_{\infty}\leq 1$ (see \cite[Supplementary material]{dirksen2021covariance}), we can make the potentially crude estimate
\begin{equation}
\begin{aligned}
\lambda_{\min}(\Cm_{\rv}) &\geq \lambda_{\min}(\diag(\Cm_{\rv})) - \|\operatorname{offdiag}(\Cm_{\rv})\| \\ 
&\geq 1-\|\operatorname{offdiag}(\Cm_{\yv})\|,
\end{aligned}
\end{equation}
so that it is sufficient if 
\begin{equation}
\|\operatorname{offdiag}(\Cm_{\yv})\|\leq 1-\theta.
\end{equation}
Note that the latter condition also implies \eqref{eqn:nondiagBoundLemma}.
Finally, let us comment on the condition linking $\eps_{\infty}$ and $\eps_{\sf F}$ in \eqref{eqn:epsinftyepsF}. In the application that follows, we will see that the $\ell_{\infty}$-error achieved by the estimator $\widehat{\Cm}_{\yv}$ is a factor $M$ smaller than the achieved Frobenius norm error. As a consequence, the relation between $\eps_{\infty}$ and $\eps_{\sf F}$ will be satisfied.  \hfill $\lozenge$
\end{remark}

\subsection{Channel covariance estimation from quantized samples}
In this part, we present an approach to estimate the covariance matrix $\Cm_{\yv}$ from a finite number of samples so that we can use the estimate $\widehat{\Cm}_{\yv}$ to apply the plug-in BLMMSE channel estimator in \eqref{eq:mmseEstCov}. Assume that the BS collects $N$ unquantized i.i.d. samples $\{\yv_n\}_{n=1}^N$ for covariance estimation\footnote{We assume that the window of $N$ samples is designed such that the samples are sufficiently spaced in the time-frequency domain, resulting in statistically independent channel snapshots. At the same time, the whole window spans a time significantly shorter than the geometry coherence time $T_{\rm c}$, so that the APS remains unchanged.} and applies coarse quantization in the ADCs. In the case of a spatially WSS channel, the diagonal of $\Cm_{\hv}$ is constant and the (non-dithered) one-bit samples $\rv^{\mathrm{nd}}_n$ defined in \eqref{eq:r-nd} can be used. Defining the sample covariance of the quantized samples
\begin{equation}
    \widehat{\Cm}_{\rv}^{\mathrm{nd}} = \frac{1}{N}\sum_{n=1}^N\rv^{\mathrm{nd}}_n \left(\rv^{\mathrm{nd}}_n\right)^\herm,
\end{equation}
the true covariance matrix $\Cm_{\yv}$ can then be estimated via the arcsine-law \cite{bar2002doa, jacovitti1994estimation,dirksen2021covariance}
\begin{equation} 
\label{eq:OneBitEstimatorND}
    \widehat{\Cm}_{\yv}^{\mathrm{nd}} = \sin\left(\frac{\pi}{2}\Re\left(\widehat{\Cm}_{\rv}^{\mathrm{nd}}\right)\right) + j \sin \left(\frac{\pi}{2}\Im\left(\widehat{\Cm}_{\rv}^{\mathrm{nd}}\right)\right).
\end{equation}
Due to the spatially non-WSS property in our model, however, it is seen from the formulation in \eqref{eq:channelCov_true} that the channel covariance may have a non-constant diagonal and non-Toeplitz structure. In such a scenario, the estimator in \eqref{eq:OneBitEstimatorND} will perform poorly since it enforces a constant diagonal.

\corrS{To overcome this limitation of the quantizer $\operatorname{csign}$, we will introduce random dithering \cite{GrN98,GrS93,Rob62}. The beneficial effect of dithering in memoryless one-bit quantization was recently rigorously analyzed in the context of one-bit compressed sensing, see, e.g., \cite{BFN17,Dir19,dirksen2023robust,DiM18a,JMP19,KSW16}. 
We will adopt a covariance estimator from \cite{dirksen2021covariance} that uses dithered quantized samples with uniformly distributed dithers}. Specifically, we assume that the real and imaginary parts of each entry are quantized independently with two independent dithers, so that we are given four dithered one-bit samples
\begin{align} 
\label{eq:rd}
    &\left\{ \Re(\rv_n^{\text{d}}), \Im(\rv_n^{\mathrm{d}}), \Re(\widetilde{\rv}_n^{\mathrm{d}}), \Im(\widetilde{\rv}_n^{\mathrm{d}}) \right\} := \\
    & \qquad \Big\{ \sign\big(\Re(\yv_n) + \tauv_n^{\Re}\big), 
    \sign\big(\Im(\yv_n) + \tauv_n^{\Im}\big), \nonumber\\
    & \;\;\qquad\sign\big(\Re(\yv_n) + \widetilde{\tauv}_n^{\Re}\big), \sign\big(\Im(\yv_n) + \widetilde{\tauv}_n^{\Im}\big) \Big\},\nonumber
\end{align}
where the real dithering vectors $\tauv_n^{\Re}, \tauv_n^{\Im},\widetilde{\tauv}_n^{\Re},\widetilde{\tauv}_n^{\Im} \in \mathbb{R}^M$, for $n \in [N]$, are independent and uniformly distributed in $[-\lambda,\lambda]^M$ and $\lambda>0$ is a tuning parameter. An example of an implementation of such a dithered quantization design is illustrated in Fig.~\ref{fig:dither}. The real part and imaginary part of the received signal after the radio frequency (RF) circuits are sampled and stored separately in two sample-and-hold (S/H) circuits. Then, a switch is used to extract in turn the signals from two S/H circuits and forward them to the one-bit ADC. Meanwhile, a dithering signal generated by the dithering generator (DG) is added into the one-bit ADC and the analog signal is dithered quantized. For instance, if the switches connect the $a$ points, the DG generates random dithering signals $\tau^{\text{Re}}$ and $\tau^{\text{Im}}$. On the other hand, if the $b$ points are connected, $\widetilde{\tau}^{\text{Re}}$ and  $\widetilde{\tau}^{\text{Im}}$ are generated from DG. The quantized signals of all antenna chains will be processed in the digital signal processor (DSP). Note that we use S/H circuits to avoid using two one-bit ADCs for each real or imaginary part signal. Also, this circuit can be directly used for non-dithered one-bit quantization by fixing the connection of switches and turning off the DG.  
\begin{figure}
	\begin{center}
	\includegraphics[width=.9\columnwidth]{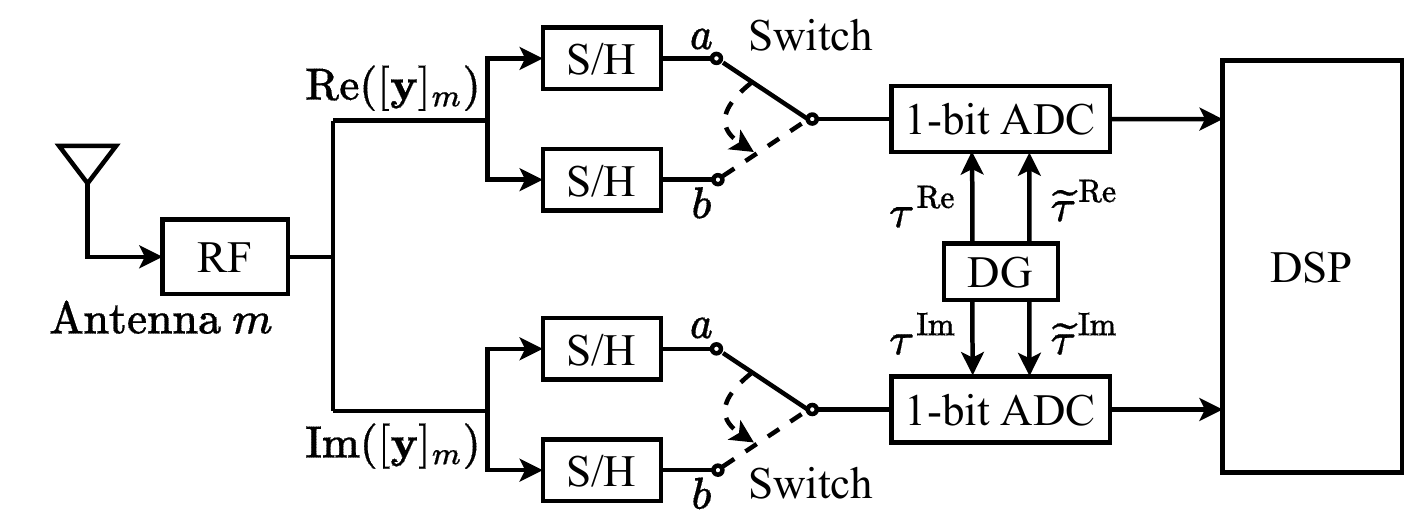}
	\caption{Illustration of the implementation of dithered one-bit quantizer in the $m$-th antenna chain.}\label{fig:dither}
	\end{center}
 \vspace{-.5cm}
\end{figure}

Given $N$ dithered quantized samples from \eqref{eq:rd}, we can estimate $\Cm_{\yv}$ via
\begin{align} \label{eq:TwoBitEstimator}
    \widehat{\Cm}^{\mathrm{d}}_{\yv} &= \frac{1}{2}\widetilde{\Cm}^{\mathrm{d}} + \frac{1}{2}\left(\widetilde{\Cm}^{\mathrm{d}}\right)^\herm, \\
\text{where} \quad
    \widetilde{\Cm}^{\mathrm{d}} &= \frac{\lambda^2}{N}\sum_{n=1}^N \rv_n^{\mathrm{d}}  \left(\widetilde{\rv}_n^{\mathrm{d}} \right)^\herm\label{eq:AsymmetricEstimator}
\end{align}
is an asymmetric version of the sample covariance matrix scaled with $\lambda^2$. 
Regarding the specific shape of the estimator in \eqref{eq:AsymmetricEstimator}, note that \cite[Lemma 15]{dirksen2021covariance} shows that if a random vector $\mathbf y$ takes values in $[-\lambda,\lambda]^M$ and $\boldsymbol{\tau}_n,\widetilde{\boldsymbol{\tau}}_n$ are uniformly distributed on $[-\lambda,\lambda]^M$, then
\begin{equation}    \mathbb{E}_{\boldsymbol{\tau}_n,\widetilde{\boldsymbol{\tau}}_n}\left[\lambda^2\sign(\mathbf y_n + \boldsymbol{\tau}_n)\sign(\mathbf y_n + \widetilde{\boldsymbol{\tau}}_n)^\transp\right]=\mathbf y_n\mathbf y_n^\transp.
\end{equation}
Since the argument carries over to our complex-valued setting in a straightforward way, $\widetilde{\mathbf C}^{\rm d}$ is an unbiased estimator in this bounded case. If $\mathbf y$ is unbounded but concentrates around its mean (as we assume in the present work), then $\widetilde{\mathbf C}^{\rm d}$ is biased, but the bias can be controlled by setting $\lambda$ large enough, cf. Lemma \ref{lem:linftyBiasEst} in Appendix \ref{sec:AppendixFrobeniusBound} and also \cite[Lemma 16]{dirksen2021covariance}.

We can now quantify the approximation of $\Cm_{\yv}$ by $\widehat{\Cm}^{\mathrm{d}}_{\yv}$ for all random vectors $\yv \in \mathbb C^M$ with $S$-subgaussian coordinates.
\begin{definition}
    We say that a random vector $\yv \in \mathbb{C}^M$ with covariance matrix $\Cm_{\yv}$ has \emph{$S$-subgaussian coordinates} if, for all $p\geq 2$ and $j\in [M]$,
\begin{align}\label{eq:SubgaussianEntries}
   & \max\left\{\Big(\mathbb{E}\big[|[\Re(\yv)]_j|^p\big]\Big)^{\frac{1}{p}},\Big(\mathbb{E}\big[|[\Im(\yv)]_j|^p\big]\Big)^{\frac{1}{p}}\right\}\nonumber\\
   &\quad\leq S\sqrt{p}\|\Cm_{\yv}\|_{\infty}^{\frac{1}{2}}.
\end{align}
\end{definition}
Note that if $\yv \in \mathbb{C}^M$ is complex Gaussian with mean zero, then both $\Re(\yv)$ and $\Im(\yv)$ are mean-zero real Gaussian vectors with covariance matrix $\tfrac{1}{2}\Re(\Cm_{\yv})$. Hence, $\yv$ has $S$-subgausian coordinates for some absolute constant $S$. \corrS{The following estimates, which complement operator norm error bounds derived in \cite{dirksen2021covariance}, are tailored to be used in Lemma~\ref{lem:Stability}.}

\begin{theorem} \label{thm:FrobeniusDithered}
Let $\yv\in \mathbb{C}^M$ be a mean-zero random vector with covariance matrix $\E\left[ \yv\yv^\H \right] = \Cm_{\yv}$ and  $S$-subgaussian coordinates. Let $\yv_1,...,\yv_N \overset{\mathrm{i.i.d.}}{\sim} \yv$. Then there exists a constant $c > 0$ which only depends on $S$ such that if $\lambda^2 \gtrsim \log(N) \| \Cm_{\yv} \|_\infty$, the covariance estimator $\widehat{\Cm}^{\mathrm{d}}_{\yv}$ fulfills, for any $t \ge 0$, with probability at least $1-8e^{-c N t}$
    \begin{align} 
        \pnorm{ \widehat{\Cm}^{\mathrm{d}}_{\yv} - \Cm_{\yv}}{\infty} &\lesssim \lambda^2 \sqrt{\frac{\log(M) + t}{N}} \\
\text{and} \quad
        \pnorm{ \widehat{\Cm}^{\mathrm{d}}_{\yv} - \Cm_{\yv}}{\sfF} &\lesssim \lambda^2 \sqrt{\frac{M^2(\log(M) + t)}{N}}.\label{eq:FrobeniusDithered}
    \end{align}
\end{theorem}
\begin{proof}
    See Appendix~\ref{sec:AppendixFrobeniusBound}.
\end{proof}
By combining Theorem~\ref{thm:FrobeniusDithered} with Lemma~\ref{lem:Stability}, we can derive a bound on the expected estimation error of the channel vector in terms of the number of samples $N$ used to estimate $\Cm_{\yv}$.

\begin{theorem} \label{thm:MainResult}
There exist constants $c_1,\ldots,c_4>0$ depending only on $S$ such that the following holds. Let $\yv \in \mathbb C^M$ be a zero-mean random vector with covariance matrix $\Cm_{\yv}$ and $S$-subgaussian coordinates. Let $\yv_1,...,\yv_N \overset{\mathrm{i.i.d.}}{\sim} \yv$. Suppose that $\Cm_{\yv}$, $\Cm_{\rv}$, and $\theta\in (0,1)$ satisfy \eqref{eqn:nondiagBoundLemma} and \eqref{eqn:lowBoundLemma}. Further suppose that $\lambda^2 \geq c_1 \log(N) \| \Cm_{\yv} \|_\infty$ and 
\begin{align}
N\geq &c_2\lambda^4 M^2 \corrS{\Big(\theta^{-6}+\theta^{-12}\|\Cm_{\hv}\|_{\sf F}^{-2} \ \max\{1,\|\Cm_{\hv}\|^2\}\|\Cm_{\yv}\|^2\Big)} \nonumber\\
&\times\max\{1,\|\Cm_{\yv}\|_{\infty}^2\}  \big(\log( M) + t\big).
\end{align}
Then, for any $t \ge 0$, with probability at least $1-8e^{-c_3 N t}$
\begin{align}
    \E \left[\left\| \widehat\hv - \widehat{\hv}^{\text{BLM}} \right\|_{2}^2\right] \leq &c_4 \lambda^2 \theta^{-6} M  \max\{1,\|\Cm_{\yv}\|_{\infty}\} \times \\ &  \noteJ{\max\{1,\|\Cm_{\hv}\|_{}\} \|\Cm_{\hv}\|_{\sf F}} \sqrt{\frac{\log(M) + t}{N}}. \nonumber
\end{align}
\end{theorem}
\begin{proof}
    See Appendix~\ref{sec:AppendixMainResult}.
\end{proof}

\begin{remark}
Theorem \ref{thm:MainResult} implies that the parameter $\lambda$ of the uniform distribution of the dithering vectors must be carefully tuned. Unfortunately, there is no closed-form method to extract an optimal $\lambda$ from the data. 
Nevertheless, Theorem \ref{thm:FrobeniusDithered} yields sufficient conditions on the choice of $\lambda$ in terms of the number of samples $N$ and the maximum variance of the data $\| \mathbf C_{\mathbf{y}} \|_\infty$, since Theorem \ref{thm:FrobeniusDithered} holds whenever $\lambda^2 \gtrsim \log(N) \| \mathbf C_{\mathbf{y}} \|_\infty$. 
In practice, one would estimate $\| \mathbf C_{\mathbf{y}} \|_\infty$ from the data in advance. This can be done more easily than estimating the full covariance matrix $\mathbf C_{\mathbf{y}}$ but is not yet covered by the present analysis. Note that in the context of unlimited sampling \cite{eamaz2022uno}, the dithering parameter of the one-bit quantizer is related to the $\ell_\infty$-norm of the \emph{bounded} samples. The sufficient condition in Theorem \ref{thm:FrobeniusDithered} is an analogous condition for our setting with samples that are not uniformly bounded --- the choice of $\lambda$ now depends on the maximal variance and the number of samples instead of on the $\ell_\infty$-norm of the samples.
\hfill $\lozenge$
\end{remark}

\subsection{APS-based channel covariance estimation}
\label{sec:APSCovEst}

Let us now revisit the problem of estimating the channel covariance based on an estimator $\widehat{\Cm}_{\yv}$ of $\Cm_{\yv}$. Previously we used the basic estimator for the channel covariance given by
\begin{align}\label{eq:basic_estimator_Ch}
    \widehat{\Cm}_{\hv} = \left(\widehat{\Cm}_{\yv} - N_0\mathbf{I}_M\right),
\end{align}
which may not necessarily be a positive semi-definite matrix. To heuristically improve the performance of this basic estimator, we further exploit the angle domain and apply a commonly considered APS-based covariance fitting to estimate the APS and subsequently enhance the estimate of the channel covariance, see e.g., \cite{khalilsarai2020structured,khalilsarai2018fdd,khalilsarai2021dual,yang2023structured}. Specifically, 
assuming that the visible antennas of all scattering clusters are known at the BS, i.e., the BS has the exact knowledge of 
the selection matrices $\{\Sm_i\}_{i=1}^L$. Using $G$ Dirac delta functions that are equally spaced in the angle domain with AoAs $\{\theta_i\}_{i=1}^G$  as dictionaries, the channel covariance matrix can be approximated as 
\begin{align}
\label{eq:parametric_covariance}
    \widetilde{\Cm}_\hv(\widehat{\gammav}) = \widetilde{\Am} \diag(\widehat{\gammav}_{L+1})\widetilde{\Am}^\herm + \sum^L_{i=1}\Sm_i\widetilde{\Am} \diag(\widehat{\gammav}_i)\widetilde{\Am}^\herm\Sm_i^\herm,
\end{align}
where $\widetilde{\Am} := [\av(\theta_1), \dots, \av(\theta_G)]$ and we define $\widehat{\gammav} \in\mathbb{R}^{\widetilde{G}}_+ :=[\widehat{\gammav}_1^\transp, \dots, \widehat{\gammav}_{L+1}^\transp]^\transp$ as the non-negative coefficients to be estimated, where $\widetilde{G} = (L+1)G$. Then, we estimate the coefficients by fitting the parametric channel covariance $\widetilde{\Cm}_\hv(\widehat{\gammav})$ to the basic estimator $\widehat{\Cm}_{\hv}$ in terms of the Frobenius norm. We denote the estimated coefficients by
\begin{align}
\label{eq:frobenius_norm}
    \widehat{\gammav}^{\star}=\argmin_{\widehat{\gammav} \in\mathbb{R}_+^{\widetilde{G}}}\; \left\|\widetilde{\Cm}_\hv(\widehat{\gammav}) - \widehat{\Cm}_{\hv}\right\|_{\sf F}^2.
\end{align}
By defining $\bv := \vec(\widehat{\Cm}_{\hv})$, $\Bm := \left[\Bm^{\text{l}}_1, \dots, \Bm^{\text{l}}_L, \Bm^{\text{c}}\right]$, and $\bA_i := \av(\theta_i)\av(\theta_i)^\herm$, where 
\begin{align}
    \Bm^{\text{c}} &:= \left[\vec(\bA_1) , \dots, \vec(\bA_G)\right], \\ 
    \Bm^{\text{l}}_i &:= \left[\vec(\Sm_i\bA_1\Sm_i^\herm) , \dots, \vec(\Sm_i\bA_G\Sm_i^\herm)\right],  \; \forall i \in [L],
\end{align}
we can rewrite \eqref{eq:frobenius_norm}  as a NNLS problem, given as 
\begin{align}
\label{eq:NNLS}
\widehat{\gammav}^{\star} = \argmin_{\widehat{\gammav} \in\mathbb{R}_+^{\widetilde{G}}}\; \left\|\Bm \widehat{\gammav} - \bv\right\|^2_2.
\end{align}
This problem can be efficiently solved using a variety of convex optimization techniques (see, e.g., \cite{chen2010nonnegativit,lawson1974solving}). In our simulations, we use the novel MATLAB NNLS solver \cite{bill2022nnls} which is much faster and stabler than the built-in MATLAB function \textit{lsqnonneg}, especially under large problem dimensions as considered here. The analysis of the computational complexity of the considered NNLS problem in \eqref{eq:NNLS} can be found in \cite[Appendix C]{yang2023structured}. With the estimated APS $\widehat{\gammav}^{\star}$ in hand, we obtain $\widehat{\Cm}_\hv^{\star}:=\widetilde{\Cm}_\hv(\widehat{\gammav}^{\star})$ as the final estimator of the channel covariance. 

\section{Data transmission rate}
\label{sec:receiver}

In the UL data transmission phase, we assume $K$ users simultaneously transmit their data. The received signal at the BS before quantization is given by 
\begin{align}
    \yv_{\text{D}} = \Hm\sv + \nv_{\text{D}},
\end{align}
where $\Hm \in \mathbb{C}^{M\times K}$ is the channel matrix of $K$ users, $\sv\sim\mathcal{CN}(\mathbf{0},\mathbf{I}_K)$ is the vector of the data signals of $K$ users and $\nv_{\text{D}}\sim\mathcal{CN}(\mathbf{0},N_0\mathbf{I}_M)$ is additive white Gaussian noise. Note that we use the subscript `D' to indicate the signal during data transmission. After the one-bit non-dithered quantization, the quantized signal is given by
\begin{align}\label{eq:bussgang_data}
    \rv_{\text{D}} = \mathcal{Q}(\yv_{\text{D}}) = \Am_{\rm D}\Hm\sv + \Am_{\rm D}\nv_{\text{D}} + \qv_{\text{D}} = \widetilde{\Hm}\sv + \widetilde{\nv}_{\rm D},
\end{align}
where $\Am_{\rm D} = \sqrt{\frac{2}{\pi}}(\diag(\Hm\Hm^\herm+N_0\mathbf{I}_M))^{-\frac{1}{2}}$ is the Bussgang gain calculated similarly as in \eqref{eq:Am} and \eqref{eq:Am_fianl}, $\widetilde{\Hm}:=\Am_{\rm D}\Hm$ is the effective channel and $\widetilde{\nv}_{\rm D} := \Am_{\rm D}\nv_{\text{D}} + \qv_{\text{D}}$ is the total noise. By treating interference as noise, we apply a linear receiver $\Wm^\herm$ to separate the quantized signal into $K$ streams as
\begin{align}
    \widehat{\sv} &= \Wm^\herm\rv_{\text{D}} = \Wm^\herm\Am_{\rm D}\Hm \sv + \Wm^\herm(\Am_{\rm D}\nv_{\text{D}} + \qv_{\text{D}}).
\end{align}
The data signal of the $k$-th user is then decoded by the $k$-th element of $\widehat{\sv}$ as 
\begin{align}
    \widehat{s}_k &= \wv_k^\herm\Am_{\rm D}\hv_k s_k + \wv_k^\herm\sum^K_{i\neq k}\Am_{\rm D}\hv_i s_i +  \wv_k^\herm(\Am_{\rm D}\nv_{\text{D}} + \qv_{\text{D}}),\nonumber
\end{align}
where $\wv_k$ and $\hv_k$ are the $k$-th columns of the $\Wm$ and $\Hm$, respectively. The covariance matrix of the statistically equivalent quantizer noise $\qv_{\text{D}}$ is given by
\begin{equation}
    \Cm_{\qv_{\rm D}} = \Cm_{\rv_{\rm D}} - \Am_{\rm D}\Cm_{\yv_{\rm D}}\Am_{\rm D}^\herm,
\end{equation}
where $\Cm_{\yv_{\rm D}} = \Hm \Hm^\herm + N_0\mathbf{I}_M$ and the covariance matrix $\Cm_{\rv_{\rm D}}$ of $\rv_{\text{D}}$ can be obtained via the arcsine law as $\Cm_{\rv_{\rm D}}= \mathcal{P}_{\text{arcsine}}(\Cm_{\yv_{\rm D}})$. Note that the quantizer noise $\qv_{\text{D}}$ is non-Gaussian. Considering the worst case by treating $\qv_{\text{D}}$  as Gaussian distributed with the same covariance matrix $\Cm_{\qv_{\rm D}}$, we can obtain a lower bound of the optimistic ergodic sum rate\footnote{Note that the ``optimistic ergodic sum rate’’ is an upper bound assuming Gaussian signaling since it assumes that the useful signal coefficient and the interference variance are perfectly known \cite{caire2018ergodic}.} of $K$ users \cite{diggavi2001worst}, given in \eqref{eq:sum-rate}.

\begin{figure*}[!b]
\normalsize
\vspace*{-3mm}
\hrulefill
\begin{align}
\label{eq:sum-rate}
    R_{\text{sum}} = \sum_{k=1}^K\mathbb{E}\left[\log_2\left(1+\frac{|\wv_k^\herm\Am_{\rm D}\hv_k|^2}{\sum^K_{i\neq k}|\wv_k^\herm\Am_{\rm D}\hv_i|^2+ N_0 \|\wv_k^\herm\Am_{\rm D}\|^2_2 + \wv_k^\herm \Cm_{\qv_{\text{D}}} \wv_k }\right)\right]
\end{align} 
\vspace*{-1cm}
\end{figure*}

Now, we consider the design of the linear receiver $\Wm$. Using the proposed plug-in channel estimator, the channel matrix $\Hm$ is estimated as $\widehat{\Hm}$. Then, the conventional MRC and ZF receivers are given as 
\begin{align}\label{eq:receivers_mrc_zf}
    \Wm^\herm_{\text{MRC}} = \widehat{\Hm}^\herm,\quad
    \Wm^\herm_{\text{ZF}} = \left(\widehat{\Hm}^\herm\widehat{\Hm}\right)^{-1}\widehat{\Hm}^\herm.
\end{align}
Note that the conventional MRC and ZF receivers do not consider the effective channel and might not perform so well due to the quantized signal. Therefore, we propose a BLMMSE receiver that takes the quantized signal directly into account by considering the effective channel and total noise in the Bussgang decomposition  \eqref{eq:bussgang_data}, which is expected to yield better performance. Specifically, similar to the BLMMSE channel estimator in  \eqref{eq:mmse}, the BLMMSE receiver is given as
\begin{align}
    \Wm^\herm_{\text{BLM}} &= \Cm_{\sv\rv_{\rm D}}\Cm_{\rv_{\rm D}}^{-1}, \\ \text{where} \quad\quad
    \Cm_{\sv\rv_{\rm D}} &= \bE\left[\sv\rv_{\rm D}^\herm\right] = \Hm^\herm\Am_{\rm D}^\herm, \\ 
    \Cm_{\rv_{\rm D}} &= \bE\left[\rv_{\rm D}\rv_{\rm D}^\herm\right] = \mathcal{P}_{\text{arcsine}}(\Cm_{\yv_{\rm D}}).
\end{align}
In practice, using the estimated channel matrix $\widehat{\Hm}$, the BLMMSE receiver $\Wm^\herm_{\text{BLM}}$ is given as
\begin{equation}
\label{eq:receivers_blm}
    \Wm^\herm_{\text{BLM}} = \widehat{\Hm}^\herm\widehat{\Am}_{\rm D}^\herm \left(\mathcal{P}_{\text{arcsine}}\left(\widehat{\Cm}_{\yv_{\rm D}}\right)\right)^{-1},
\end{equation}
where $\widehat{\Am}_{\rm D} = \sqrt{\frac{2}{\pi}}(\diag(\widehat{\Hm}\widehat{\Hm}^\herm+N_0\mathbf{I}_M))^{-\frac{1}{2}}$ and $\widehat{\Cm}_{\yv_{\rm D}}=\widehat{\Hm}\widehat{\Hm}^\herm+N_0\mathbf{I}_M$.

\section{Simulation results}
\label{sec:numerics}

In our simulation, we take $M = 256$ antennas at the BS in ULA. The channel consists of $L^{\text{c}} = 3$ multipaths in common clusters and $L=2$ local clusters, where each local cluster is composed of three multipaths, i.e., $L^{\text{l}}_i=3, \; i = 1, 2$. The first local cluster is visible to the first quarter of antennas and the second local cluster is visible to the last quarter of antennas, i.e., $\mathcal{M}^{\text{l}}_1 = \{1, 2, \dots, \frac{M}{4}\}, \; \mathcal{M}^{\text{l}}_2 = \{\frac{3M}{4} + 1, \frac{3M}{4} + 2, \dots, M\}$. The AoAs of the common clusters and the first and second local clusters are uniformly and randomly generated from $[-60, 60]$, $[-60,0]$ and $[0,60]$ degrees, respectively, i.e., $\theta^{\rm c}_i \sim \mathcal{U}(-60,60), \forall i \in [L^{\rm c}], \theta^{\rm l}_{1,j} \sim \mathcal{U}(-60,0)  ,\forall j \in [L^{\rm l}_1], \theta^{\rm l}_{2,j} \sim \mathcal{U}(0,60) ,\forall j \in [L^{\rm l}_2]$. The APS of all multipaths are randomly generated with the constraints $P^{\text{c}} = 0.3$, $P^{\text{l}}_1 = 0.7$, and $P^{\text{l}}_2 = 0.5$. Note that this setting satisfies the assumption of $\max(\diag(\Cm_{\hv})) = 1$. Unless otherwise stated, the SNR is set to $10$dB (equivalently, $N_0=0.1$). We consider three different basic estimators $\widehat{\Cm}_{\yv}$ for $\Cm_{\yv}$ and correspondingly $\widehat{\Cm}_\hv$ for $\Cm_\hv$ according to \eqref{eq:basic_estimator_Ch} (denoted as ``Basic''): the estimator \eqref{eq:OneBitEstimatorND}, based on non-dithered one-bit quantized samples, the estimator \eqref{eq:TwoBitEstimator}, based on dithered one-bit quantized samples, and finally, as a reference, the sample covariance matrix of the \emph{unquantized} samples. We then use the results of basic estimators to produce the 
APS-fitting-based estimator $\widehat{\Cm}_\hv^{\star}$ (denoted as ``NNLS''), as is detailed in Section~\ref{sec:APSCovEst}. 
All results presented below are averaged over 10 random channel geometry realizations, each with 20 groups of $N$ i.i.d. random channel realizations.  

\subsection{Channel covariance estimation} 
Given an estimator $\widehat{\Cm}_\hv^{\star}$ of the channel covariance matrix, we first evaluate it in terms of the normalized Frobenius norm error, which is given by 
\begin{equation}
E_{\rm NF} = \E\left[\frac{\|\Cm_\hv - \widehat{\Cm}_\hv^{\star}\|^2_{\sf F}}{\|\Cm_\hv \|^2_{\sf F}}\right].
\end{equation}
The numerical results of the proposed dithered one-bit estimator under different values of $\lambda$ are shown in Fig.~\ref{fig:E_NF_lambda}. It is seen that the choice of $\lambda$ significantly influences the results of dithered quantization. Although the resulting $E_{\rm NF}$ of the basic estimator is very sensitive to $\lambda$, the APS-based method can significantly enhance not only the covariance estimation accuracy but also the robustness against $\lambda$, especially under a larger number of samples $N$.
Moreover, the results of the three estimators under various $N$ are shown in Fig.~\ref{fig:E_NF_N}. For the APS-based estimators, we provide results of two different grid sizes ($G=M$ and $G = 2M$) in the angle domain. While the APS-based method improves estimation for all basic estimators, we observe that all APS-based results are saturated for large $N$. Especially for the extreme unquantized case under $N>100$, the APS-based estimator with $G=M$ shows a larger $E_{\rm NF}$ than the basic estimator but the $E_{\rm NF}$ becomes much smaller after increasing the grid size to $G=2M$. This indicates that the grid size $G$ would be the bottleneck for large $N$. Nevertheless, for the one-bit case, $G=M$ already brings a notable improvement for a large range of $N$, and the resulting $E_{\rm NF}$ of the proposed dithered one-bit scheme is much smaller than the $E_{\rm NF}$ of the non-dithered one-bit scheme for all $N$. 

\begin{figure*}[t]
\centering
        \begin{subfigure}[b]{0.45\textwidth}
        \includegraphics[width=\columnwidth]{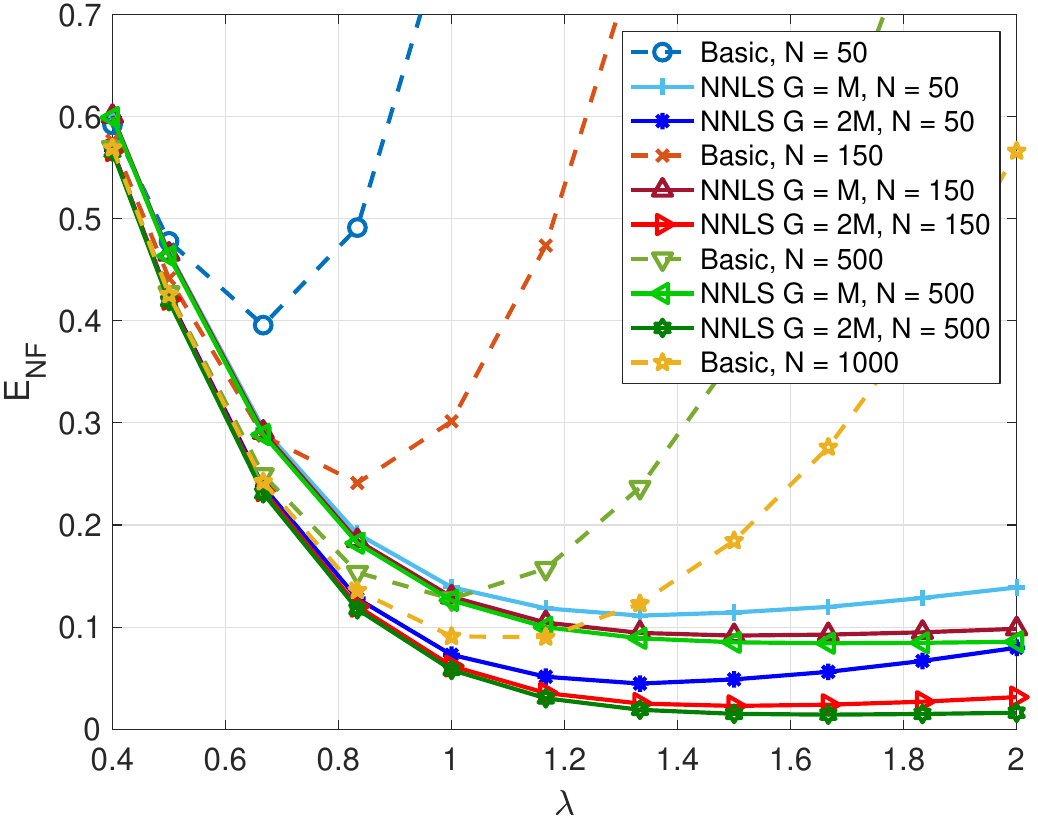}
        \caption{$E_{\rm NF}$ v.s. $\lambda$ of proposed dithered one-bit estimator}
        \label{fig:E_NF_lambda}
        \end{subfigure}
        ~
        \begin{subfigure}[b]{0.45\textwidth}
        \includegraphics[trim={0 1.1mm 0 0mm}, width=1.02\columnwidth]{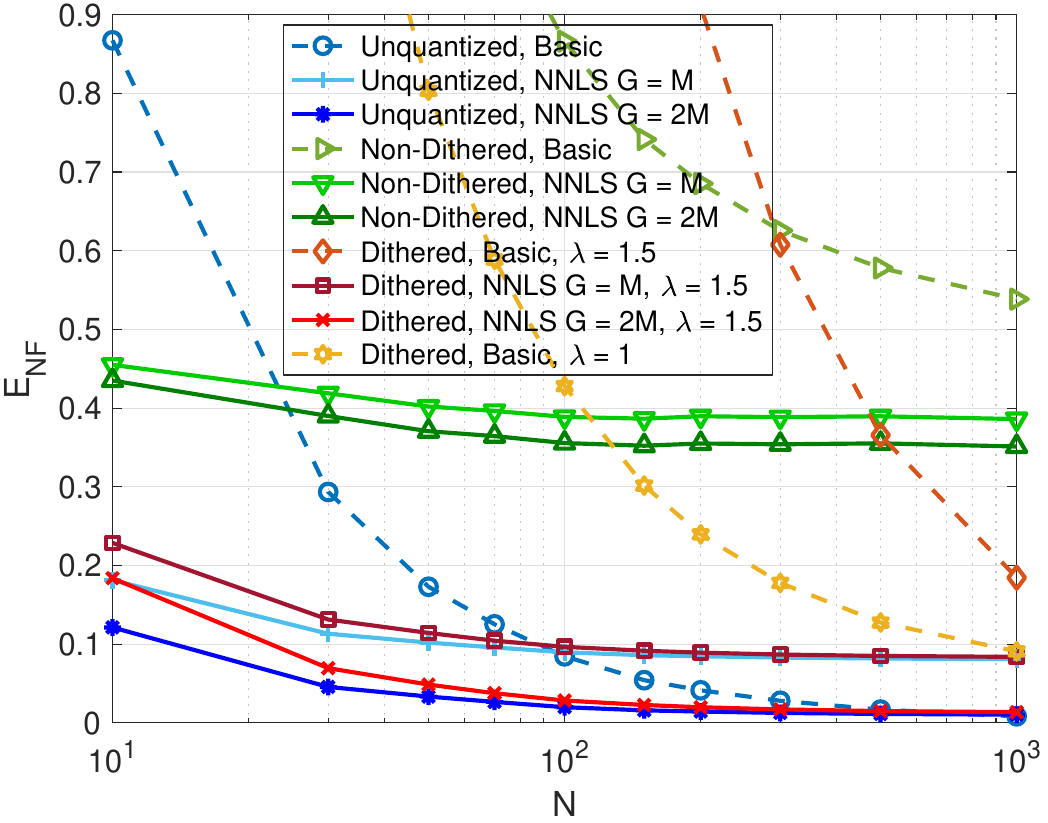}
        \caption{$E_{\rm NF}$ v.s. $N$ of all estimators}
        \label{fig:E_NF_N}
        \end{subfigure}
	\caption{$E_{\rm NF}$ under various $\lambda$ in (a) and $N$ in (b).}
	    \label{fig:E_NF}
     \vspace{-.3cm}
\end{figure*} 

\subsection{Channel vector estimation via BLMMSE}

Next, we numerically evaluate the BLMMSE-based channel vector estimator in terms of the normalized MSE, which is given by 
\begin{equation}
E_{\rm NMSE} = \frac{\E\left[\|\hv - \widehat{\hv}\|^2_2\right]}{\trace(\Cm_\hv)}.
\end{equation}
Given an estimated channel covariance, we calculate the NMSE with 100 i.i.d. random channel realizations. The averaged results 
under various $\lambda$ and $N$ are depicted in Fig.~\ref{fig:E_NMSE_lambda} and Fig.~\ref{fig:E_NMSE_N}, respectively, where the lower bound is obtained by using the true channel covariance. Since the basic channel covariance estimators without APS fitting are not guaranteed to generate PSD matrices (especially under small $N$), their results may not be directly applicable to channel estimation. We therefore only present channel estimation results based on APS fitting. 
It is observed again that the choice of $\lambda$ significantly influences the estimation performance of the dithered case. By applying a proper $\lambda$ (e.g., $\lambda=1$ in Fig.~\ref{fig:E_NMSE_N}) the channel estimate can be considerably improved compared to the non-dithered case for a large range of number of samples. Moreover, it is seen from Figs.~\ref{fig:E_NF_lambda} and \ref{fig:E_NMSE_lambda} that the trend of tuning $\lambda$ in BLMMSE based channel estimation is different from the trend in channel covariance estimation. In the range of $1\leq\lambda\leq2$, the results of BLMMSE-based channel estimation are no longer as robust as the results in covariance estimation even under large $N$. The optimal choices of $\lambda$ for the two estimation problems are also different, e.g., $\lambda^\star \approx 1.5$  in Fig.~\ref{fig:E_NF_lambda} and $\lambda^\star\approx 1.2$ in Fig.~\ref{fig:E_NMSE_lambda} under $N=500$.

\begin{figure*}[t]
\centering
        \begin{subfigure}[b]{0.45\textwidth}
        \includegraphics[width=\columnwidth]{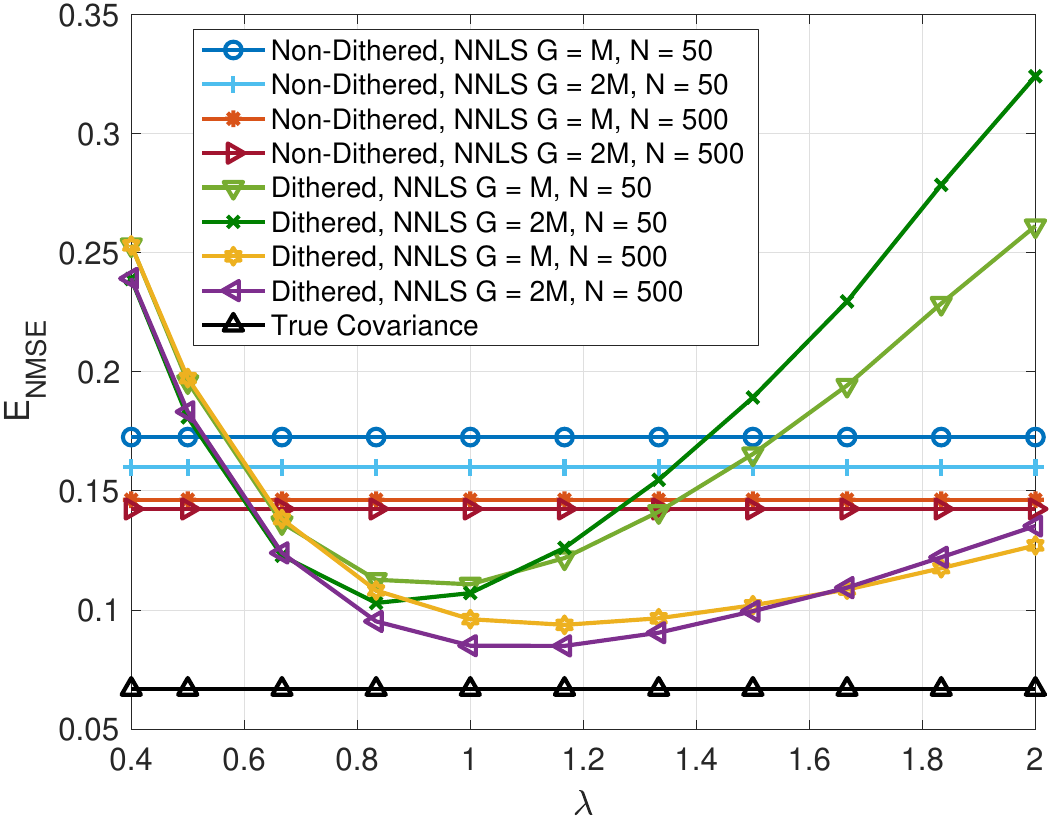}
        \caption{$E_{\rm NMSE}$ v.s. $\lambda$}
        \label{fig:E_NMSE_lambda}
        \end{subfigure}
        ~
        \begin{subfigure}[b]{0.45\textwidth}
        \includegraphics[trim={0 1.1mm 0 0mm},width=1.02\columnwidth]{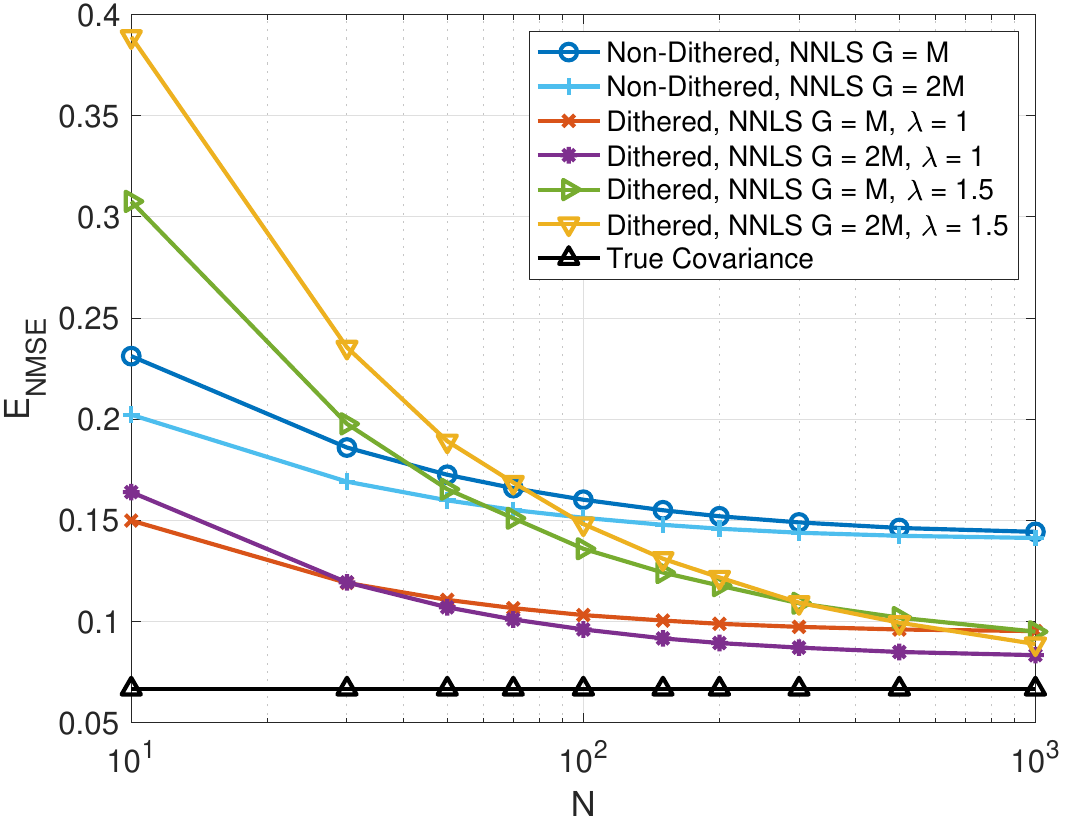}
        \caption{$E_{\rm NMSE}$ v.s. $N$}
        \label{fig:E_NMSE_N}
        \end{subfigure}
	\caption{$E_{\rm NMSE}$ via BLMMSE channel estimator under various $\lambda$ in (a) and  $N$ in (b). }
	    \label{fig:E_MMSE}
     \vspace{-.3cm}
\end{figure*}

\subsection{Effect of SNR on channel estimation}
We now evaluate the effect of the SNR on channel covariance and channel vector estimation, which is shown in Fig.~\ref{fig:SNR}. From Fig.~\ref{fig:E_NF_SNR} and Fig.~\ref{fig:E_NMSE_SNR} we observe that, unlike the trend of unquantized samples and the proposed dithered one-bit samples, where their $E_{\rm NF}$ and $E_{\rm NMSE}$  monotonically decrease with the increase of the SNR, the resulting estimation errors of non-dithered samples start increasing when the SNR is larger than 5dB. To understand the counter-intuitive behavior in the non-dithered case, note that we can alternatively view the noise as adding dithering to the channel vector before quantization. Hence, as we have already seen that appropriate random dithering is beneficial for recovering the covariance, random noise with a proper power would actually help to recover the covariance.  This observation confirms the benefit of our design with dithering for the considered spatially non-stationary system.

\begin{figure*}[t]
\centering
        \begin{subfigure}[b]{0.45\textwidth}
        \includegraphics[width=\columnwidth]{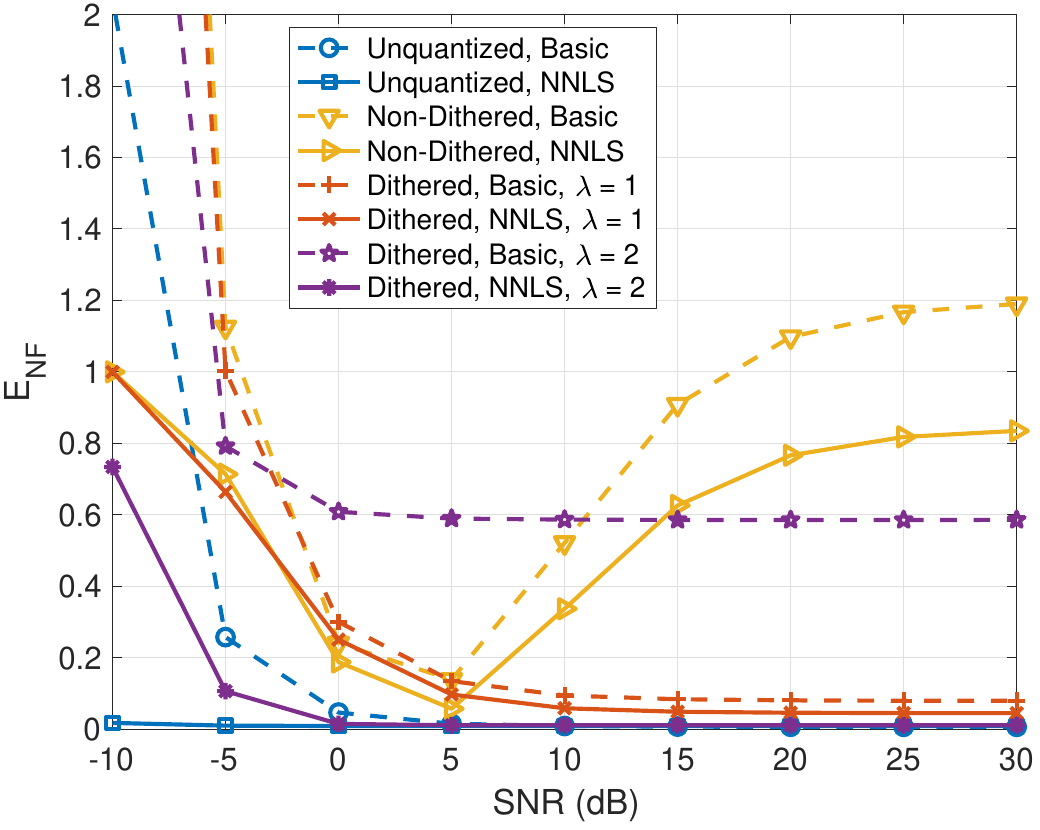}
        \caption{$E_{\rm NF}$ v.s. SNR}
        \label{fig:E_NF_SNR}
        \end{subfigure}
        ~
        \begin{subfigure}[b]{0.45\textwidth}
        \includegraphics[width=\columnwidth]{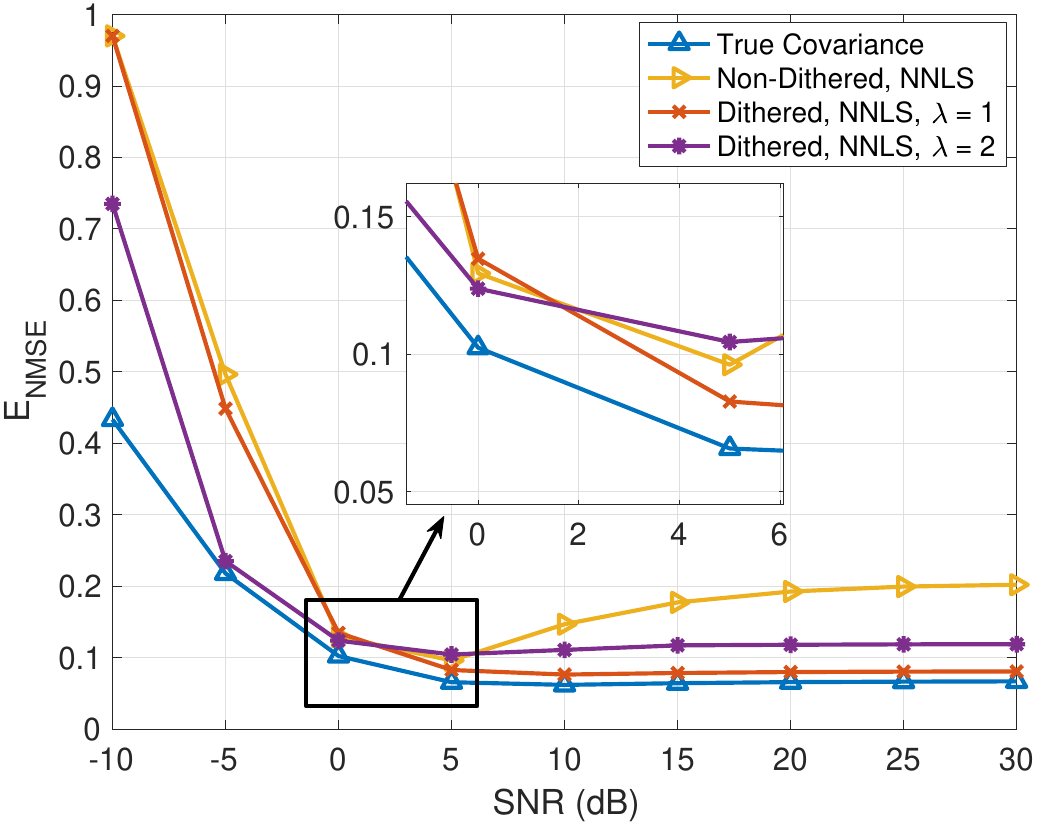}
        \caption{$E_{\rm NMSE}$ v.s. SNR}
        \label{fig:E_NMSE_SNR}
        \end{subfigure}
	\caption{$E_{\rm NF}$ in (a) and $E_{\rm NMSE}$ in (b) under various SNR with $N=1000$ and $G=2M$.}
	    \label{fig:SNR}
     \vspace{-.3cm}
\end{figure*}

\subsection{Ergodic sum rate evaluation}
Given the estimated channels, we evaluate the proposed scheme in terms of the ergodic sum rate given in \eqref{eq:sum-rate}. We test with $K = 4$ users and assume that the channel geometry of all users follows the setting described at the beginning of this section. For each channel estimation, we test with MRC, ZF, and BLMMSE receivers given in \eqref{eq:receivers_mrc_zf}, and \eqref{eq:receivers_blm}. The results of non-dithered and dithered schemes are all based on APS-fitting with $G=M$.  Similarly as in the previous part, besides the non-dithered and dithered schemes we also provide results based on the true channel covariance matrix. However, unlike for the channel MSE criterion used in the previous part, the use of the true covariance is not guaranteed to yield a better sum rate, as a channel estimator with a smaller MSE does not necessarily yield a higher sum rate. Furthermore, we provide results based on the true channel vectors as upper bounds.

We first present the resulting sum rates under various $\lambda$ with $N=50$ samples for covariance estimation and BLMMSE channel estimation in Fig.~\ref{fig:rate_lambda_all}. It is first observed that the BLMMSE receiver performs much better than the ZF and MRC receivers. Especially note that, compared to the ZF and MRC receivers with perfect CSI,  the BLMMSE receiver with one-bit quantized samples still shows significant performance gain. This is expected since the BLMMSE receiver takes the quantization into account whereas the conventional ZF and MRC receivers do not. Next, we observe that the results based on the true covariance do not always provide the largest sum rate as previously explained.  Specifically, among the results with ZF receivers (the dashed lines) the non-dithered one is the best. Since the ZF and MRC receivers are designed to deal with non-quantized received signals and perform much worse than the BLMMSE receiver, we now focus on the results of the BLMMSE receiver, which is also individually depicted in Fig.~\ref{fig:rate_lambda_BLMMSE} with two more cases of $N=500$ and $N=1000$ samples. It is noticed again that a larger number of samples $N$ makes the results with dithering in terms of the sum rate not only better but also more robust against $\lambda$. From Fig.~\ref{fig:rate_lambda_BLMMSE} it is seen that the highest sum rate obtained by the proposed scheme with dithering (under $N=1000$ and $\lambda \approx 0.6$) is very close to the result based on the true covariance matrix. This shows the advantage of the proposed scheme for both channel estimation and multi-user receivers under one-bit quantization.

Finally, we focus on the influence of the number of samples $N$. The sum rates under various $N$ of MRC, ZF, and BLMMSE are depicted in Fig.~\ref{fig:rate_N_all} and of only BLMMSE with three different choices of $\lambda$ for the dithered case are depicted in Fig.~\ref{fig:rate_N_BLMMSE}. In Fig.~\ref{fig:rate_N_all} we see a similar behavior as before: the BLMMSE receiver produces better sum rates than the ZF and MRC receivers over a large range of $N$. It is seen from Fig.~\ref{fig:rate_N_BLMMSE} that under the best $\lambda^\star \approx 0.6$ the proposed scheme with dithering produces sum rates comparable to the results based on the true channel covariance when $N\geq 100$. It is additionally observed from Fig.~\ref{fig:rate_N_BLMMSE} that as the number of samples $N$ increases the difference between the results obtained with different $\lambda$ is decreasing. This indicates again that under larger $N$ the results with dithering are more robust to variations in $\lambda$. 

\begin{figure*}[t]
\centering
        \begin{subfigure}[b]{0.45\textwidth}
        \includegraphics[width=\columnwidth]{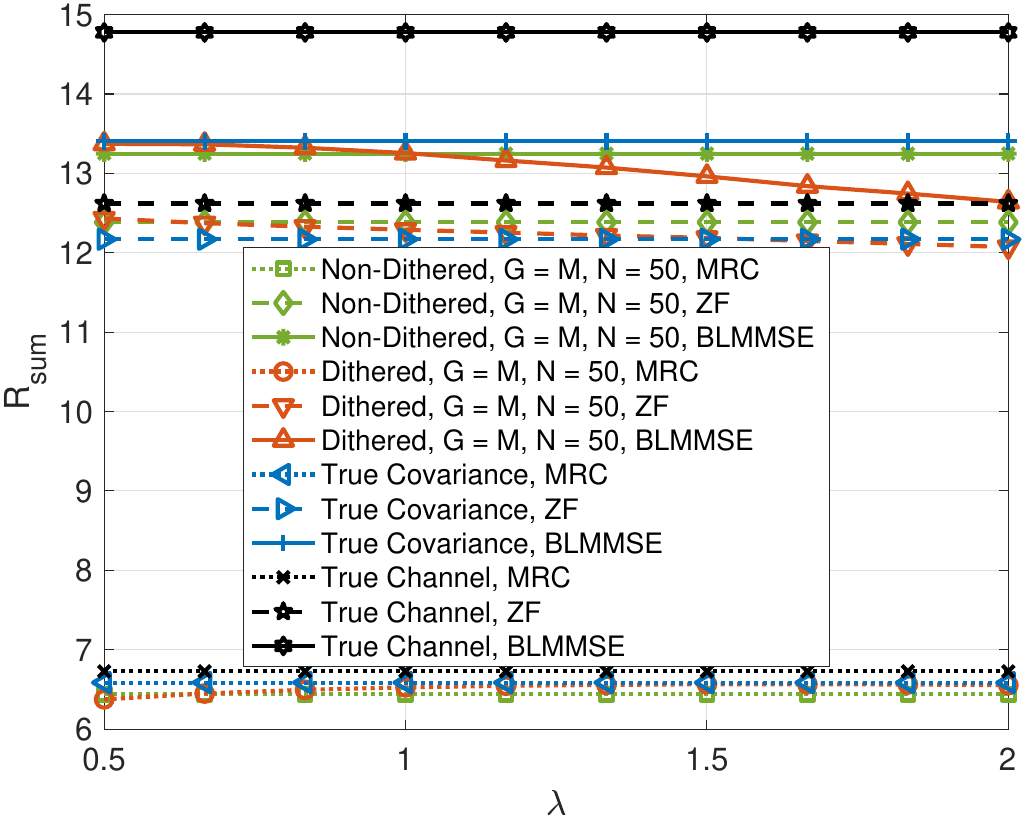}
        \caption{$R_{\rm sum}$ v.s. $\lambda$ via MRC, ZF, BLMMSE receivers}
        \label{fig:rate_lambda_all}
        \end{subfigure}
        ~
        \begin{subfigure}[b]{0.45\textwidth}
        \includegraphics[width=1.025\columnwidth]{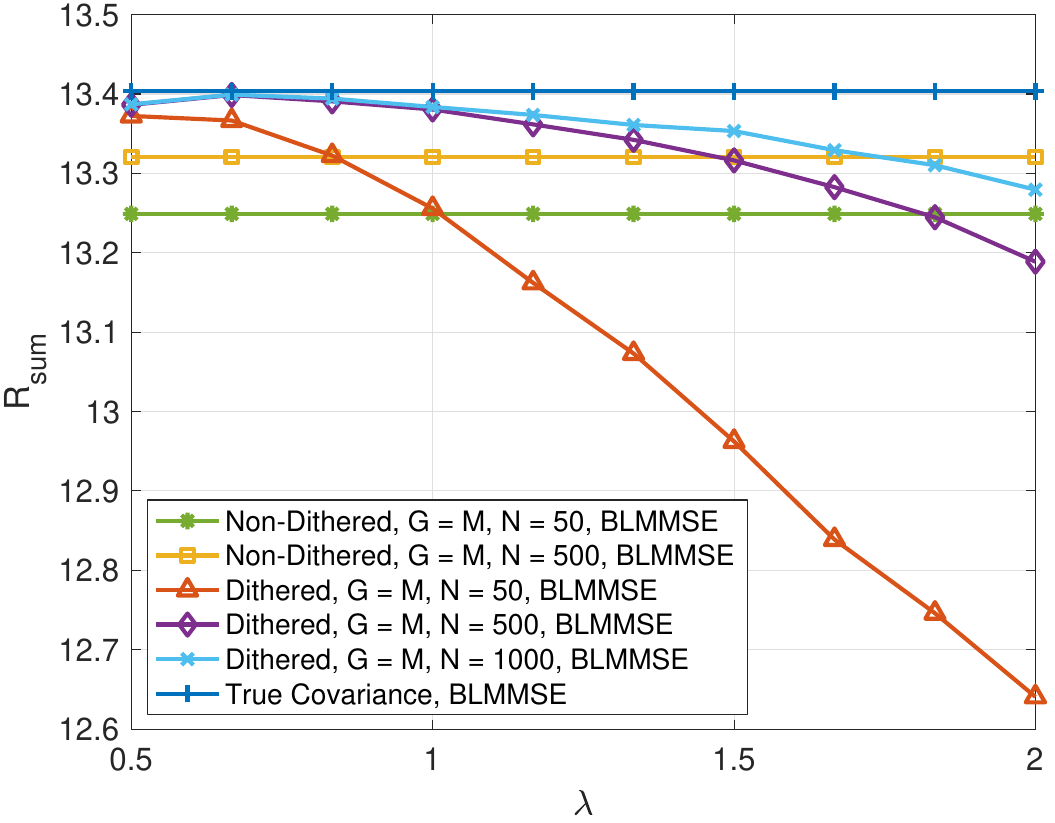}
        \caption{$R_{\rm sum}$ v.s. $\lambda$ via BLMMSE receiver}
        \label{fig:rate_lambda_BLMMSE}
        \end{subfigure}
        ~
        \begin{subfigure}[b]{0.45\textwidth}
        \includegraphics[width=\columnwidth]{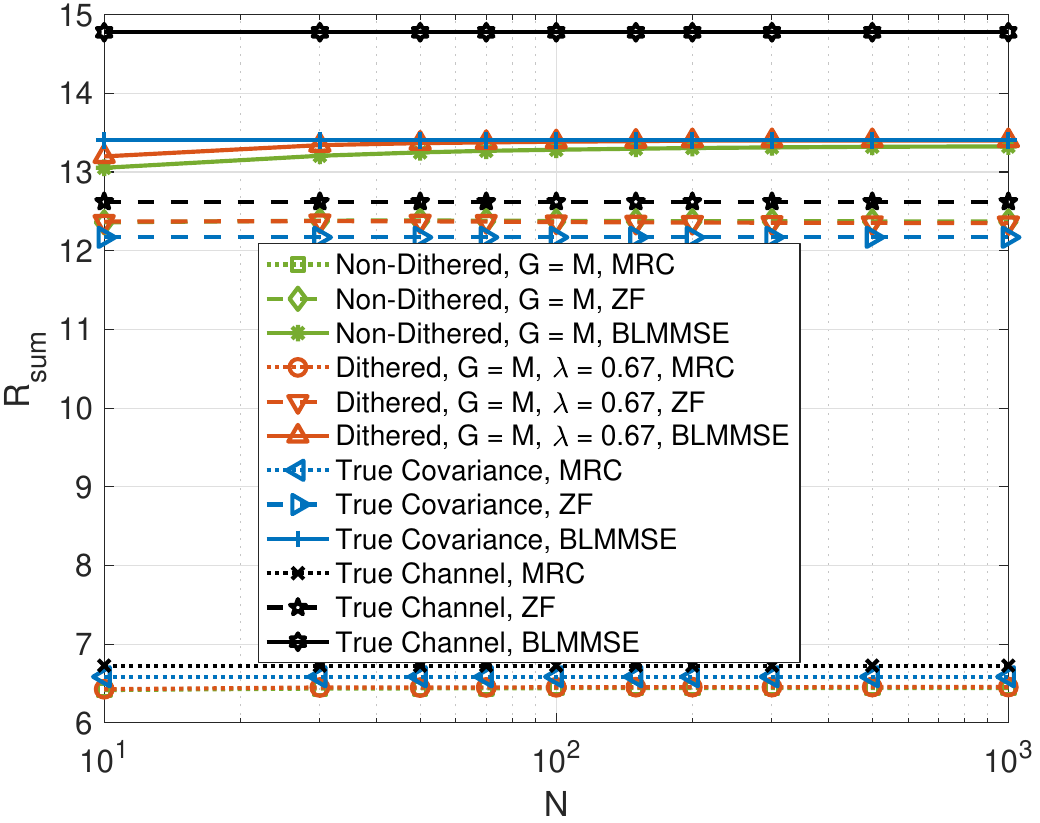}
        \caption{$R_{\rm sum}$ v.s. $N$ via MRC, ZF, BLMMSE receivers}
        \label{fig:rate_N_all}
        \end{subfigure}
        ~
        \begin{subfigure}[b]{0.45\textwidth}
        \includegraphics[width=1.025\columnwidth]{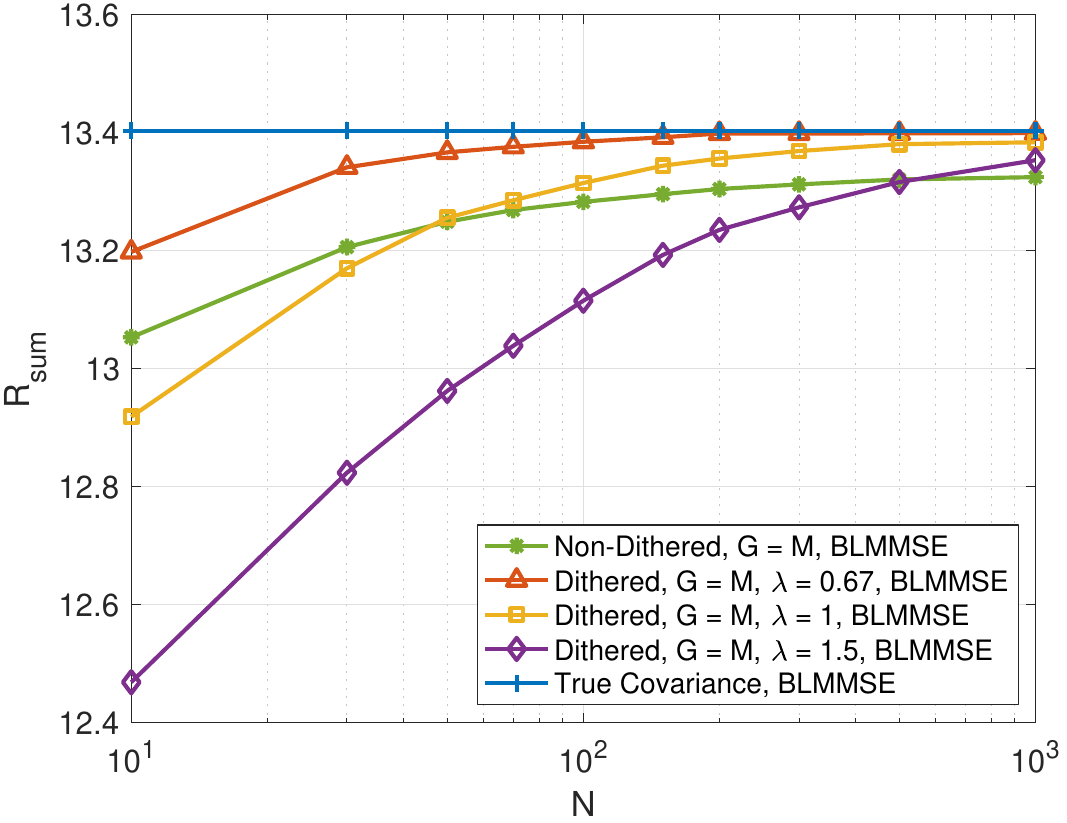}
        \caption{$R_{\rm sum}$ v.s. $N$ via BLMMSE receiver}
        \label{fig:rate_N_BLMMSE}
        \end{subfigure}
	\caption{$R_{\rm sum} $ of $K=4$ users via all tested receivers in (a), (c) and enlarged view of results via BLMMSE receiver in (b), (d). }
	    \label{fig:sum-rate}
     \vspace{-.3cm}
\end{figure*}

\subsection{Comparison to conventional systems with multi-bit ADC}
In order to further justify the applicability of the proposed system with dithered one-bit ADCs, we provide a comparison to conventional systems with multi-bit ADCs in terms of both sum rate and energy efficiency. We suppose that a conventional system equipped with multi-bit ADCs uses the quantized signal as if they are unquantized and applies basic empirical covariance estimation and plug-in MMSE channel estimation. We assume that the ADC type is flash that has high conversion speed and consists of $2^B-1$ comparators for a $B$ bits ADC \cite{walden1999analog}. The power consumption of a $B$ bits flash ADC is thus given by $P_{\rm ADC}(B) = (2^B-1) P_{\rm COMP}$, where $P_{\rm COMP}$ is the operational  power of a single comparator. For the quantization design, we assume that all ADCs have uniform quantization levels\footnote{Uniformly spaced quantization levels are used in most commercial ADCs. Moreover, for a Gaussian signal under relatively low resolutions (e.g., 3-5 bits) the optimized uniform levels only lead to marginal MSE difference compared to optimal arbitrary quantization levels, see Fig. 1 in \cite{max1960quantizing}.} and the quantization levels are identically optimized for Gaussian signal with zero-mean and unit variance according to Table 1 in \cite{desset2015validation}. The flash ADC has an automatic gain control (AGC) mechanism that scales the input signal to the dynamic range of the ADC. We consider a simple case that all ADCs scale the input by $\sqrt{\frac{2}{1+N_0}}$ according to the SNR in \eqref{eq:received_signal}.

The ergodic rate for multi-bit cases requires the analysis of quantization error, which is not a  straightforward extension of the one-bit result. Therefore, we assume that the received signal in the data transmission phase is not quantized and thus the ergodic sum rate can be simplified to the conventional form as 

\begin{equation}
    R^{\rm con}_{\rm sum} = \sum^K_{k=1}\bE\left[\log_2\left(1+\frac{|\wv^\herm_k \hv_k|^2}{\sum^K_{i\neq k}|\wv^\herm_k \hv_i|^2 + N_0\|\wv_k\|^2_2}\right)\right].
\end{equation}

The energy efficiency (EE) is defined as $EE(B) = \frac{R^{\rm con}_{\rm sum}(B)}{2 M P_{\rm ADC}(B)}$, where the factor $2$ in the denominator is due to separate quantization for in-phase and quadrature signal. Note that for each input signal, the proposed design uses the dithered one-bit ADC \textit{twice} for channel covariance estimation and \textit{once} for plug-in channel estimation. Therefore, the general power consumption of the proposed dithered one-bit ADC scheme can be given  as $2 M \omega P_{\rm ADC}(1)$, where $\omega \in [1,2]$ takes account for twice uses of ADC in covariance estimation. Since we only care about the trend of EE with $B$, we focus on the normalized EE denoted as $NEE(B) = EE(B) / EE(1)$.

In this simulation, we set $N=10^4$ to avoid inaccurate results due to a lack of samples. The results of the proposed scheme with dithered one-bit ADC are based on NNLS enhanced with $G=2M$ and a manual search of best $\lambda$. The resulting  $E_{\rm NF}$ and $E_{\rm NMSE}$ are depicted in Fig.~\ref{fig:E_multiBit}. Using ZF receiver, the resulting sum rate $R^{\rm con}_{\rm sum}$ and $NEE$ of the conventional system and proposed scheme under three values of $\omega$ are depicted in Fig.~\ref{fig:EE}. We observe from Fig.~\ref{fig:EE} that the system performance in terms of sum-rate $R^{\rm con}_{\rm sum}$ of the proposed dithered one-bit scheme is much better than the conventional one-bit system, while still worse than the conventional two-bit system due to the channel estimation with non-dithered one-bit signal. Nevertheless, the results of $NEE$ show that the proposed dithered one-bit scheme has a higher energy efficiency even with $\omega=2$ compared to the conventional two-bit system. In practice, the number of assigned resource blocks for channel covariance estimation should be significantly less than the number for instantaneous channel estimation, thus resulting in an $\omega$ very close to one and consequently much higher energy efficiency compared to the conventional system. Furthermore, while the conventional system under higher ADC resolutions can provide a significantly higher sum rate, the system energy efficiency decays also dramatically due to the fact that the power consumption of a flash ADC grows exponentially with its resolution.

\begin{figure*}[t]
\centering
        \begin{subfigure}[b]{0.45\textwidth}
        \includegraphics[width=1.03\columnwidth]{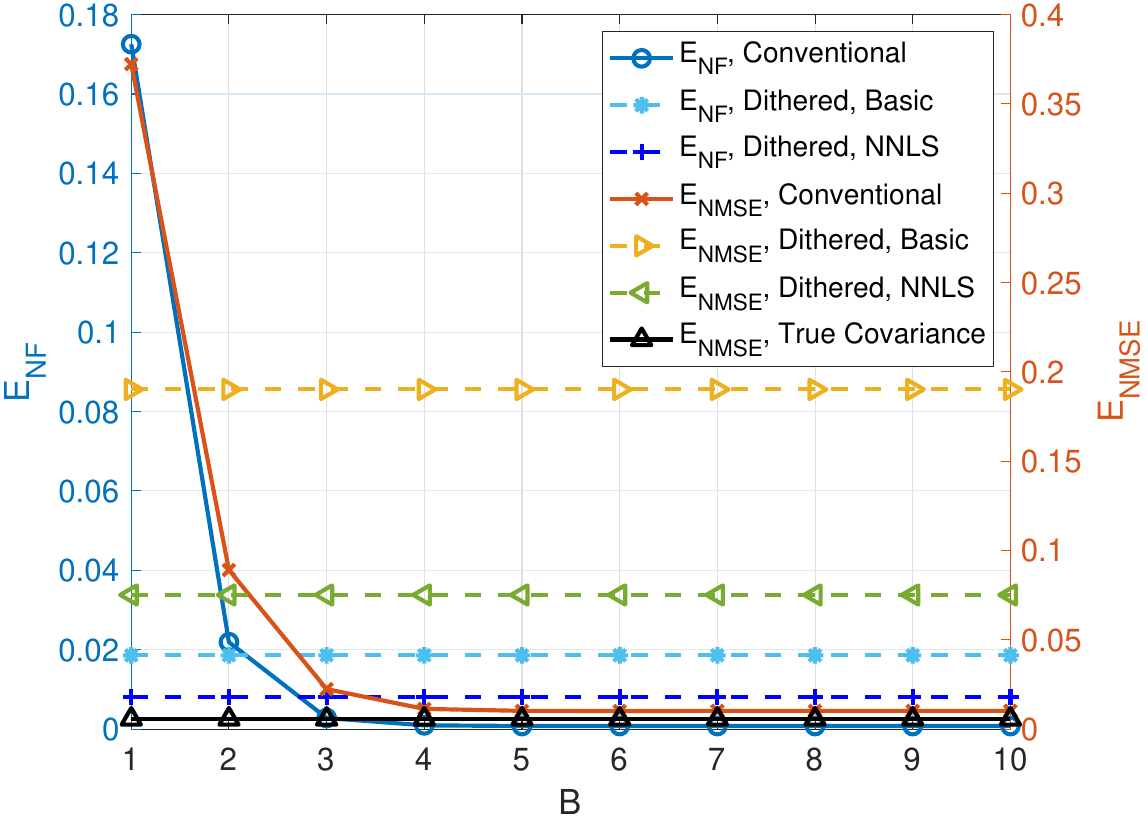}
        \caption{$E_{\rm NF}$ and $E_{\rm NMSE}$ v.s. $B$}
        \label{fig:E_multiBit}
        \end{subfigure}
        ~
        \begin{subfigure}[b]{0.45\textwidth}
        \includegraphics[width=\columnwidth]{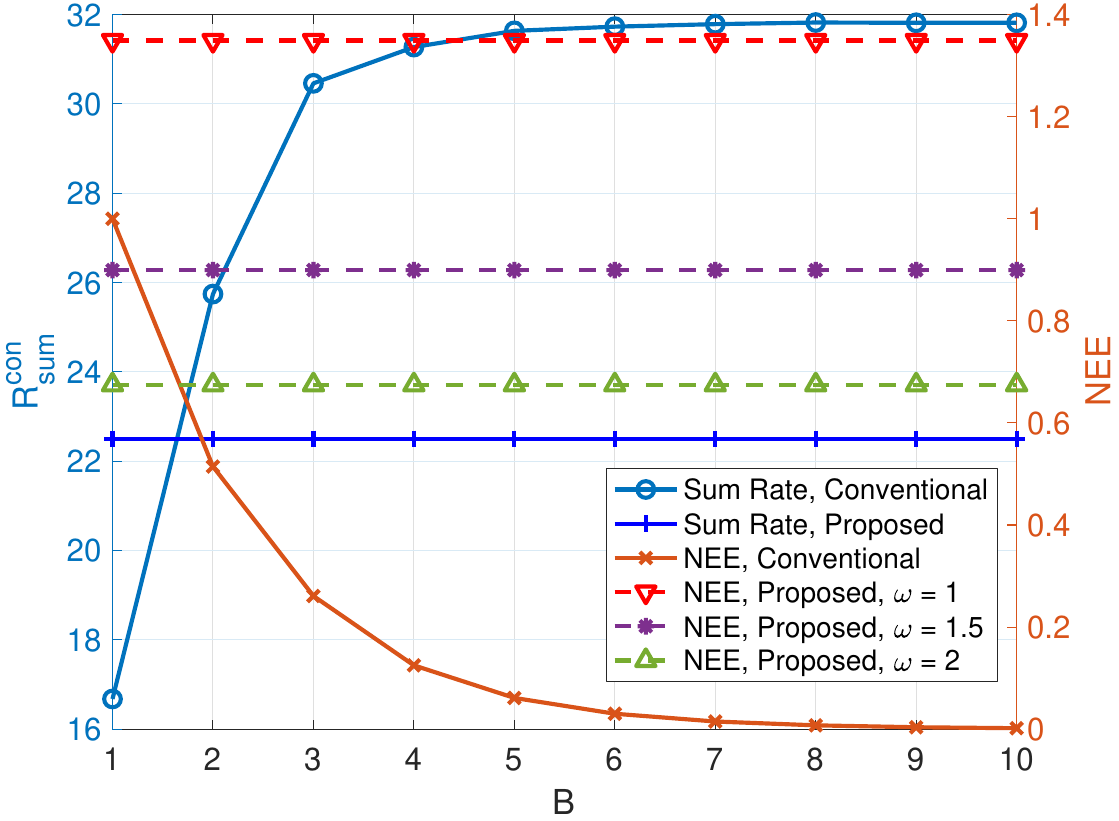}
        \caption{$R^{\rm con}_{\rm sum}$ and $NEE$ v.s. $B$ via ZF receiver}
        \label{fig:EE}
        \end{subfigure}
	\caption{Comparison to the conventional multi-bit system under $N=10^4$ and $G=2M$.}
	    \label{fig:multiBit}
     \vspace{-.3cm}
\end{figure*}

\section{Conclusion and discussion}
\label{sec:discussion} 

In this work, we proposed a plug-in channel estimator for massive MIMO systems with spatially non-stationary channels and one-bit quantizers. We analyzed the quantized signal via the Bussgang decomposition and analyzed the distortion produced by using an estimated, rather than the true, channel covariance in the construction of the BLMMSE estimator of the channel. To obtain an estimate of the covariance of the spatially non-stationary channel, we introduced a channel covariance estimator based on dithered quantized samples and theoretically analyzed its performance. We further enhanced this estimator using an APS-based NNLS solution. Our numerical results showed large performance gains of the proposed scheme with dithering in terms of both channel vector and covariance estimation. Finally, we proposed a BLMMSE-based receiver tailored to one-bit quantized data signals for multi-user data transmission phase and showed in numerical experiments that it outperforms the conventional MRC and ZF receivers in terms of the ergodic sum rate. 

There are two important aspects of our work that can be improved. First, we observed in the numerical experiments that the hyperparameter $\lambda$ of the dithering generation influences the channel estimation significantly. Even though this influence was observed to diminish as the sample size increases, it is still of significant interest to develop a data-driven method to optimally tune $\lambda$. Second, in the proposed APS-based channel covariance estimation scheme, the visibility of local clusters is assumed to be known at the BS. In practice, however, the visibility of local clusters is usually not easy to estimate. It is therefore desirable to develop a scheme without the assumption of known visibility of local clusters. We will investigate these two questions in future work. 

\appendices

\section{Proof of Lemma \ref{lem:Stability}}
\label{sec:AppendixStability}
Recall that
\begin{align}
    \widehat{\hv}^{\text{BLM}} 
    &= \Cm_{\hv\rv} \Cm_{\rv}^{-1} \rv 
    =  \Cm_{\hv} \Am^\herm \Cm_{\rv}^{-1} \rv = \noteJ{(\Cm_{\yv}-N_0 \id)} \Am^\herm \Cm_{\rv}^{-1} \rv  \nonumber \\ 
    \widehat\hv &= \noteJ{\left(\widehat{\Cm}_{\yv}-N_0 \id\right)} \widehat \Am^\H \widehat\Cm_{\rv}^{-1} \rv, \nonumber
\end{align}
where $\widehat \Am$ and $\widehat\Cm_{\rv}$ are defined like $\Am$ and $\Cm_{\rv}$ with $\Cm_{\yv}$ being replaced by $\widehat\Cm_{\yv}$. Let us abbreviate $\widehat{\hv}^{\text{BLM}} = \Mm \rv$ and $\widehat\hv = \widehat\Mm \rv$.
Consider $\alpha,\beta,\gamma>0$  such that $\min_{i\in [M]}|[\Cm_{\yv}]_{i,i}| \ge \alpha$,  $\lambda_{\min}(\Cm_{\rv}) \ge \gamma$, and
\begin{align}\label{eqn:nondiagBound}
    \left| \left[\diag(\Cm_{\yv})^{-\frac{1}{2}} \Cm_{\yv} \diag(\Cm_{\yv})^{-\frac{1}{2}}\right]_{i,j} \right| \le 1-\beta,
    \; \forall i \neq j. 
\end{align}
In particular, $\theta\leq \min\{\alpha,\beta,\gamma\}$. We start by writing
\begin{align}
    & \;\E \left[\pnorm{\widehat\hv - \widehat{\hv}^{\text{BLM}}}{2}^2\right] \nonumber \\
    = & \; \E\left[ \trace \round{ \widehat{\hv}^{\text{BLM}} \left(\widehat{\hv}^{\text{BLM}}\right)^\H - \widehat{\hv}^{\text{BLM}} \widehat\hv^\H - \widehat\hv \left(\widehat{\hv}^{\text{BLM}}\right)^\H + \widehat\hv \widehat\hv^\H }\right] \nonumber\\
    =& \;\trace \round{ \Mm \Cm_{\rv} \left(\Mm - \widehat\Mm\right)^\H + \widehat\Mm \Cm_{\rv} \left(\widehat\Mm - \Mm \right)^\H }\nonumber \\
    =& \;\inner{\Mm\Cm_{\rv} + \widehat\Mm \Cm_{\rv}, \left(\Mm - \widehat\Mm\right)}{\sfF} \nonumber\\
    =& \;2\inner{\Mm\Cm_{\rv}, \left(\Mm - \widehat\Mm\right)}{\sfF} - \inner{\left(\widehat\Mm-\Mm\right) \Cm_{\rv}, \left(\Mm - \widehat\Mm\right)}{\sfF}  \nonumber\\
    \leq& \;2\|\Mm\Cm_{\rv}\|_{\sf F} \ \|\Mm - \widehat\Mm\|_{\sf F}+ \|\Cm_{\rv}\| \ \|\Mm - \widehat\Mm\|_{\sf F}^2. \nonumber
\end{align}
\noteJ{Observe that $\| \Am \| \le \frac{1}{\sqrt{\alpha}}$ by assumption such that
\begin{equation}
    \|\Mm\Cm_{\rv}\|_{\sf F} =  \| \Cm_{\hv} \Am^\herm\|_{\sf F}\leq \| \Cm_{\hv} \|_{\sf F}   \|\Am^\herm\| \leq \frac{1}{\sqrt{\alpha}} \|\Cm_{\hv}\|_{\sf F}. \nonumber
\end{equation}
}
Moreover, using that $\|\arcsin(\Bm)\|\leq \frac{\pi}{2}\|\Bm\|$ if $\|\Bm\|_{\infty}\leq 1$ (see \cite[Supplementary Material, Eq.\ (4)]{dirksen2021covariance}), we find 
\begin{align}
   \pnorm{\Cm_{\rv}}{} 
   &\le \pnorm{\diag(\Cm_{\yv})^{-\frac{1}{2}} \Re(\Cm_{\yv}) \diag(\Cm_{\yv})^{-\frac{1}{2}}}{} \nonumber\\ 
   &\quad + \pnorm{\diag(\Cm_{\yv})^{-\frac{1}{2}} \Im(\Cm_{\yv}) \diag(\Cm_{\yv})^{-\frac{1}{2}}}{} \nonumber \\
   &\le 2 \pnorm{\diag(\Cm_{\yv})^{-\frac{1}{2}} }{} \pnorm{ \Cm_{\yv} }{} \pnorm{ \diag(\Cm_{\yv})^{-\frac{1}{2}}}{} 
   \leq \frac{2}{\alpha} \pnorm{\Cm_{\yv}}{}.\nonumber
\end{align}
We conclude that 
\begin{align}
    \E \left[\pnorm{\widehat\hv - \widehat{\hv}^{\text{BLM}}}{2}^2\right] &\lesssim \alpha^{-\frac{1}{2}} \|\Cm_{\hv}\|_{\sf F} \|\Mm - \widehat\Mm\|_{\sf F} \nonumber\\ 
    &\quad+ \alpha^{-1}\pnorm{\Cm_{\yv}}{}  \|\Mm - \widehat\Mm\|_{\sf F}^2. \nonumber
\end{align}
We will show that $\|\Mm - \widehat\Mm\|_{\sf F}\leq \kappa\leq \alpha^{\frac{1}{2}}\frac{\|\Cm_{\hv}\|_{\sf F}}{\pnorm{\Cm_{\yv}}{}}$
so that we obtain $\frac{\kappa}{\sqrt{\alpha} }\|\Cm_{\hv}\|_{\sf F}$ as a final estimate. We start by estimating
\begin{align}
\label{eqn:M-hatMsplit}
    &\left\| \Mm - \widehat\Mm \right\|_{\sf F} = \noteJ{\left\| (\Cm_{\yv} - N_0 \id) \Am^\H \Cm_{\rv}^{-1} - (\widehat\Cm_{\yv} - N_0 \id) \widehat\Am^\H \widehat\Cm_{\rv}^{-1} \right\|_{\sf F} } \nonumber\\
    &\le \| \Cm_{\yv} - \widehat\Cm_{\yv} \|_{\sf F} \| \Am \| \| \Cm_{\rv}^{-1} \| + \| \noteJ{\widehat\Cm_{\yv} - N_0 \id} \|_{\sf F}\| \Am - \widehat\Am \| \| \Cm_{\rv}^{-1} \| \nonumber\\
    &\quad+ \| \noteJ{\widehat\Cm_{\yv} - N_0 \id} \|_{} \| \widehat\Am \| \| \Cm_{\rv}^{-1} - \widehat\Cm_{\rv}^{-1} \|_{\sf F}.
\end{align}
The first term is clearly bounded by $\gamma^{-1}\alpha^{-\frac{1}{2}}\eps_{\sf F}$.
To estimate the second term, we note that $\| \widehat\Cm_{\yv} - N_0 \id \|_{\sf F} \le \| \Cm_{\hv} \|_{\sf F} + \| \widehat\Cm_{\yv} - \Cm_{\yv} \|_{\sf F} \le \| \Cm_{\hv} \|_{\sf F} + \varepsilon_{\sf F}$ and $\| \widehat\Cm_{\yv} - N_0 \id \|\leq \| \Cm_{\hv} \| + \varepsilon_{\sf F}$.
Furthermore, we use that 
\begin{equation}
\label{eqn:invDist}
    \Zm_1^{-1} - \Zm_2^{-1} = \Zm_1^{-1} (\Zm_2 - \Zm_1) \Zm_2^{-1}
\end{equation}
for any invertible $\Zm_1,\Zm_2$ of the same dimensions. This yields 
\begin{align}
    &\left\|\widehat\Am -\Am \right\| \nonumber\\ &= \frac{2}{\pi} \left\|\diag(\widehat\Cm_{\yv})^{-\frac{1}{2}}\left(\diag(\Cm_{\yv})^{\frac{1}{2}}-\diag(\widehat\Cm_{\yv})^{\frac{1}{2}}\right)\diag(\Cm_{\yv})^{-\frac{1}{2}}\right\|\nonumber \\
    & \leq \frac{2}{\pi} \left\|\diag(\widehat\Cm_{\yv})^{-\frac{1}{2}}\right\|  \left \|\diag(\Cm_{\yv})^{\frac{1}{2}}-\diag(\widehat\Cm_{\yv})^{\frac{1}{2}}\right\|  
    \left\|\diag(\Cm_{\yv})^{-\frac{1}{2}}\right\|.\nonumber
\end{align}
By assumption, we have
\begin{equation}
    \left\|\diag(\Cm_{\yv})^{-\frac{1}{2}}\right\| = \sqrt{\frac{1}{\min_i [\Cm_{\yv}]_{i,i}}} \leq \sqrt{\frac{1}{\alpha}}.\nonumber
\end{equation}
Moreover, since $\|\widehat\Cm_{\yv}-\Cm_{\yv}\|_{\infty}\leq \frac{\alpha}{2}$ \noteJ{by \eqref{eqn:epsinftyepsF}}, we find 
\begin{align}
    &\min_i [\widehat\Cm_{\yv}]_{i,i} \geq \min_i [\Cm_{\yv}]_{i,i} - \left\|\widehat\Cm_{\yv}-\Cm_{\yv}\right\|_{\infty} \geq \frac{\alpha}{2} \nonumber \\ 
    \text{and so} \qquad & \left\|\diag(\widehat\Cm_{\yv})^{-\frac{1}{2}}\right\| \leq \sqrt{\frac{2}{\alpha}}.\nonumber
\end{align}
\noteJ{Note that this also implies that $\| \widehat\Am \| \lesssim \frac{1}{\sqrt{\alpha}}$. }
Using that $|\sqrt{x}-\sqrt{y}|\leq \frac{|x-y|}{\sqrt{c}}$ if $x\geq c>0$, $y\geq 0$, we find
\begin{equation}
    \left\|\diag(\Cm_{\yv})^{\frac{1}{2}}-\diag(\widehat\Cm_{\yv})^{\frac{1}{2}}\right\| \leq \sqrt{\frac{1}{\alpha}}\left\|\Cm_{\yv}-\widehat\Cm_{\yv}\right\|_{\infty}= \sqrt{\frac{1}{\alpha}}\eps_{\infty}, \nonumber
\end{equation}
and hence $\| \Am - \widehat\Am \|\leq \frac{4}{\pi}\alpha^{-\frac{3}{2}}\eps_{\infty}$.

Let us finally estimate the last term on the right-hand side of \eqref{eqn:M-hatMsplit}. Write $c_{ij}=[\Cm_{\yv}]_{i,j}$, $\hat{c}_{ij}=[\widehat \Cm_{\yv}]_{i,j}$ and observe that  
\begin{align*}
    &\left|\frac{\hat{c}_{ij}}{\sqrt{\hat{c}_{ii}\hat{c}_{jj}}}-\frac{c_{ij}}{\sqrt{c_{ii}c_{jj}}}\right|   \leq \left|\frac{\hat{c}_{ij}-c_{ij}}{\sqrt{\hat{c}_{ii}\hat{c}_{jj}}}\right| +  \frac{|c_{ij}|}{\sqrt{\hat{c}_{ii}}} \left|\frac{1}{\sqrt{\hat{c}_{jj}}} - \frac{1}{\sqrt{c_{jj}}}\right|\\ &\quad+  \frac{|c_{ij}|}{\sqrt{\hat{c}_{jj}}} \left|\frac{1}{\sqrt{\hat{c}_{ii}}} - \frac{1}{\sqrt{c_{ii}}}\right| \\
    & \lesssim \frac{1}{\alpha} \left\|\widehat\Cm_{\yv} - \Cm_{\yv}\right\|_{\infty} + \|\Cm_{\yv}\|_{\infty} \frac{1}{\alpha^2} \left\|\widehat\Cm_{\yv} - \Cm_{\yv}\right\|_{\infty} \\
    & \lesssim \|\Cm_{\yv}\|_{\infty} \frac{1}{\alpha^2} \left\|\widehat\Cm_{\yv} - \Cm_{\yv}\right\|_{\infty}\leq \frac{\beta}{2}
\end{align*}
as $\eps_{\infty} \lesssim \beta\frac{\alpha^2}{\|\Cm_{\yv}\|_{\infty}}$.
By \eqref{eqn:nondiagBound}, this implies that 
\begin{equation}
\label{eqn:nondiagBoundHat}
    \left| \left[\diag(\widehat\Cm_{\yv})^{-\frac{1}{2}} \widehat\Cm_{\yv} \diag(\widehat\Cm_{\yv})^{-\frac{1}{2}}\right]_{i,j} \right| \le 1-\frac{\beta}{2}, \; \forall i \neq j.
\end{equation}
Clearly, for any $|\arcsin(x) - \arcsin(y)| \le L_\beta |x-y|, \forall x,y \in (-1+\tfrac{\beta}{2},1-\tfrac{\beta}{2})$, where
\begin{equation*}
    L_\beta = \sup_{0\leq z<1-\tfrac{\beta}{2}} \sqrt{\frac{1}{1-z^2}} = \sqrt{\frac{1}{1-(1-\frac{\beta}{2})^2}}\leq \sqrt{\frac{2}{\beta}}.
\end{equation*}
Together with \eqref{eqn:nondiagBound} and \eqref{eqn:nondiagBoundHat} this yields
\begin{align*}
    \left\| \Cm_{\rv} - \widehat\Cm_{\rv} \right\|_{\sf F} 
    &\lesssim \beta^{-\frac{1}{2}} \Big\| \diag(\widehat\Cm_{\yv})^{-\frac{1}{2}} \Re(\widehat\Cm_{\yv}) \diag(\widehat\Cm_{\yv})^{-\frac{1}{2}} \nonumber\\ 
    & \quad - \diag(\Cm_{\yv})^{-\frac{1}{2}} \Re(\Cm_{\yv}) \diag(\Cm_{\yv})^{-\frac{1}{2}} \Big\|_{\sf F} \nonumber \\
    & + \beta^{-\frac{1}{2}} \Big\| \diag(\widehat\Cm_{\yv})^{-\frac{1}{2}} \Im(\widehat\Cm_{\yv}) \diag(\widehat\Cm_{\yv})^{-\frac{1}{2}} \nonumber \\
    &\quad- \diag(\Cm_{\yv})^{-\frac{1}{2}} \Im(\Cm_{\yv}) \diag(\Cm_{\yv})^{-\frac{1}{2}} \Big\|_{\sf F}.
\end{align*}
Now observe that 
\begin{align*}
 & \Big\| \diag(\widehat\Cm_{\yv})^{-\frac{1}{2}} \Re(\widehat\Cm_{\yv}) \diag(\widehat\Cm_{\yv})^{-\frac{1}{2}} \nonumber\\
 &\quad- \diag(\Cm_{\yv})^{-\frac{1}{2}} \Re(\Cm_{\yv}) \diag(\Cm_{\yv})^{-\frac{1}{2}} \Big\|_{\sf F} \nonumber\\
     &\le \left\| \diag(\widehat\Cm_{\yv})^{-\frac{1}{2}} - \diag(\Cm_{\yv})^{-\frac{1}{2}} \right\|  \| \widehat\Cm_{\yv} \|_{\sf F} \left\| \diag(\widehat\Cm_{\yv})^{-\frac{1}{2}} \right\| \nonumber\\ 
    &\quad  + \left\| \diag(\Cm_{\yv})^{-\frac{1}{2}} \right\|  \| \widehat\Cm_{\yv} - \Cm_{\yv} \|_{\sf F} \left\| \diag(\widehat\Cm_{\yv})^{-\frac{1}{2}} \right\| \nonumber\\
    &\quad + \left\| \diag(\Cm_{\yv})^{-\frac{1}{2}} \right\|  \| \Cm_{\yv} \|_{\sf F}  \left\| \diag(\widehat\Cm_{\yv})^{-\frac{1}{2}} - \diag(\Cm_{\yv})^{-\frac{1}{2}} \right\| \nonumber\\
    & \lesssim \alpha^{-2}\| \Cm_{\yv} \|_{\sf F}\eps_{\infty} + (\alpha^{-2}\eps_{\infty}+\alpha^{-1})\eps_{\sf F}
\end{align*}
and analogously, 
\begin{align*}
    & \Big\| \diag(\widehat\Cm_{\yv})^{-\frac{1}{2}} \Im(\widehat\Cm_{\yv}) \diag(\widehat\Cm_{\yv})^{-\frac{1}{2}}\nonumber\\
    &\quad- \diag(\Cm_{\yv})^{-\frac{1}{2}} \Im(\Cm_{\yv}) \diag(\Cm_{\yv})^{-\frac{1}{2}} \Big\|_{\sf F} \nonumber\\
    &  \lesssim \alpha^{-2}\| \Cm_{\yv} \|_{\sf F}\eps_{\infty} + (\alpha^{-2}\eps_{\infty}+\alpha^{-1})\eps_{\sf F}.
\end{align*}
Hence, 
\begin{equation*}
    \left\| \Cm_{\rv} - \widehat\Cm_{\rv} \right\|_{\sf F}\lesssim \beta^{-\frac{1}{2}}\alpha^{-2}\| \Cm_{\yv} \|_{\sf F}\eps_{\infty} + \beta^{-\frac{1}{2}}(\alpha^{-2}\eps_{\infty}+\alpha^{-1})\eps_{\sf F}.
\end{equation*}
\corrS{By our assumptions on $\eps_{\infty}$ and $\eps_{\sf F}$, the right-hand side is bounded by $\gamma/2$} and hence the assumption $ \| \Cm_{\rv}^{-1} \| \le \gamma^{-1}$ implies that  $\| \widehat\Cm_{\rv}^{-1} \| \le 2\gamma^{-1}$.
Using now again \eqref{eqn:invDist} we finally arrive at
\begin{align*}
    &\left\| \Cm_{\rv}^{-1} - \widehat\Cm_{\rv}^{-1}\right\|_{\sf F} \\
    &\quad\lesssim \beta^{-\frac{1}{2}}\gamma^{-2} \left(\alpha^{-2}\| \Cm_{\yv} \|_{\sf F}\eps_{\infty} + (\alpha^{-2}\eps_{\infty}+\alpha^{-1})\eps_{\sf F}\right).
\end{align*} 
Combining all our estimates in \eqref{eqn:M-hatMsplit}, we find
\begin{align*}
   \left \| \Mm - \widehat\Mm \right\|_{\sf F}  
&\lesssim \gamma^{-1}\alpha^{-\frac{1}{2}}\eps_{\sf F} + \gamma^{-1} \noteJ{(\| \Cm_{\hv} \|_{\sf F} + \eps_{\sf F})} \alpha^{-\frac{3}{2}}\eps_{\infty} \\
& \quad + 
\alpha^{-\frac{1}{2}}\beta^{-\frac{1}{2}}\gamma^{-2}\noteJ{(\corrS{\| \Cm_{\hv} \|} + \eps_{\sf F})} \times\\
&\qquad\left(\alpha^{-2}\| \Cm_{\yv} \|_{\sf F}\eps_{\infty} + \left(\alpha^{-2}\eps_{\infty}+\alpha^{-1}\right)\eps_{\sf F}\right).
\end{align*}
Since $\eps_{\infty}\leq \min\left\{\frac{\eps_{\sf F}}{\|\Cm_{\yv}\|_{\sf F}},1\right\}$, we can estimate the right-hand side by $\kappa := c \, \alpha^{-\frac{5}{2}}\beta^{-\frac{1}{2}}\gamma^{-2}\max\{1,\|\Cm_{\hv}\|_{}\}\eps_{\sf F}$,
for an absolute constant $c>0$. Clearly, 
\begin{equation*}
    \kappa\leq \alpha^{-\frac{1}{2}}\frac{\|\Cm_{\hv}\|_{\sf F}}{\pnorm{\Cm_{\yv}}{}}
\end{equation*}
by our assumption on $\eps_{\sf F}$, which completes the proof.

\section{Proof of Theorem \ref{thm:FrobeniusDithered}}
\label{sec:AppendixFrobeniusBound}
In the proof of Theorem \ref{thm:FrobeniusDithered} we use the following lemmas. The first one bounds the bias of \eqref{eq:AsymmetricEstimator} in terms of $\lambda$. 
\begin{lemma} \label{lem:linftyBiasEst}
    Let $S > 0$. There exist constants $c_1,c_2>0$ depending only on $S$ such that the following holds. Let $\yv\in \mathbb{C}^M$ be a mean-zero random vector with covariance matrix $\E\left[ \yv\yv^\H \right] = \Cm_{\yv}$ and $S$-subgaussian coordinates. Let $\lambda>0$ and let $\rv^\Re = \sign(\Re(\yv) + \tauv^\Re)$, $\rv^\Im = \sign(\Im(\yv) + \tauv^\Im)$, $\widetilde{\rv}^\Re = \sign(\Re(\yv) + \widetilde{\tauv}^\Re)$, and $\widetilde{\rv}^\Im = \sign(\Im(\yv) + \widetilde{\tauv}^\Im)$, where $\tauv^\Re, \tauv^\Im,\widetilde{\tauv}^\Re,\widetilde{\tauv}^\Im$ are independent and uniformly distributed in $[-\lambda,\lambda]^M$ and independent of $\yv$. Abbreviate $\rv = \rv^\Re + \i \rv^\Im$ and $\widetilde{\rv} = \widetilde{\rv}^\Re + \i \widetilde{\rv}^\Im$. Then,
    \begin{align*}
       \left\| \lambda^2 \E\left[ \rv \widetilde{\rv}^\H \right] - \Cm_{\yv} \right\|_{\infty}
       \leq c_1 (\lambda^2+\|\Cm_{\yv}\|_{\infty})e^{\frac{-c_2 \lambda^2}{\|\Cm_{\yv}\|_{\infty}}}.
    \end{align*}
\end{lemma}
\begin{proof}
    \corrS{The proof of this lemma is a straightforward extension of \cite[Lemma 17]{dirksen2021covariance} to the complex domain. We include it for the convenience of the reader.}
    First note that
    \begin{align}
    \begin{split} \label{eq:BiasBoundProof}
        &\left\| \lambda^2 \E\left[ \rv \widetilde{\rv}^{\H} \right] - \Cm_{\yv} \right\|_{\infty} \\
        &= \big\| \lambda^2 \E\left[ (\rv^\Re + \i \rv^\Im) (\widetilde{\rv}^\Re + \i \widetilde{\rv}^\Im)^\H \right] \\&\quad- \E\left[ (\Re(\yv)+\i \Im(\yv)) (\Re(\yv)+\i \Im(\yv))^\H \right] \big\|_{\infty} \\
        &\le \left\| \lambda^2 \E\left[ \rv^\Re (\widetilde{\rv}^\Re)^\transp \right] - \E\left[ \Re(\yv) \Re(\yv)^\transp \right] \right\|_{\infty}\\
        &\quad+ \left\| \lambda^2 \E\left[ \rv^\Re (\widetilde{\rv}^\Im)^\transp \right] - \E\left[ \Re(\yv) \Im(\yv)^\transp \right] \right\|_{\infty} \\
        &\quad + \left\| \lambda^2 \E\left[ \rv^\Im (\widetilde{\rv}^\Re)^\transp \right] - \E\left[ \Im(\yv) \Re(\yv)^\transp \right] \right\|_{\infty} \\
        &\quad + \left\| \lambda^2 \E\left[ \rv^\Im (\widetilde{\rv}^\Im)^\transp \right] - \E\left[ \Im(\yv) \Im(\yv)^\transp \right] \right\|_{\infty}.
    \end{split}
    \end{align}
    Since $\yv$ has $S$-subgaussian coordinates, we get from \eqref{eq:SubgaussianEntries} that $\| [\Re(\yv)]_i \|_{\psi_2}, \| [\Im(\yv)]_i \|_{\psi_2} \le S \| \Cm_{\yv} \|_\infty^\frac{1}{2}$, for any $i \in [M]$, where $\| \cdot \|_{\psi_2}$ denotes the subgaussian norm. Applying \cite[Lemma 17]{dirksen2021covariance} for $U=[\Re(\yv)]_i$ and $V=[\Re(\yv)]_j$  yields
    \begin{align*}
        &\Big| \lambda^2 \E\big[ \sign\big([\Re(\yv)]_i + [\tauv^\Re]_i \big)  \sign\big([\Re(\yv)]_j + [\widetilde{\tauv}^\Re]_j\big) \big] \nonumber\\
        &- \E\big[ [\Re(\yv)]_i [\Re(\yv)]_j \big] \Big| \lesssim (\lambda^2 + S^2 \| \Cm_{\yv} \|_\infty) e^{-c\frac{\lambda^2}{S^2 \| \Cm_{\yv} \|_\infty}}.
    \end{align*}
    Since this holds for any choice of $i,j \in [M]$, the first term on the right-hand side of \eqref{eq:BiasBoundProof} satisfies the claimed bound. The three other terms can be treated in the same way such that our claim follows.
\end{proof}
The second lemma is a simple concentration inequality that applies to dithered samples of real distributions.
\begin{lemma} \label{lem:MaxBoundDithered}
	There exist absolute constants $c_1,c_2>0$ such that the following holds. 
	\noteJ{Let $\yv,\widetilde\yv \in \mathbb R^M$ be random vectors}. Let $\yv_1,...,\yv_N \overset{\mathrm{i.i.d.}}{\sim} \yv$, \noteJ{let $\widetilde\yv_1,...,\widetilde\yv_N \overset{\mathrm{i.i.d.}}{\sim} \widetilde\yv$}, and let $\tauv_1,\dots,\tauv_N,\widetilde\tauv_1,\dots,\widetilde\tauv_N$ be independent and uniformly distributed in $[-\lambda,\lambda]$, for $\lambda > 0$. Define $\rv_k = \sign(\yv_k + \tauv_k)$ and \noteJ{$\widetilde\rv_k = \sign(\widetilde\yv_k + \widetilde\tauv_k)$}. If $N\geq c_1\log(M)$, then

	\begin{align*}
	    &\text{Pr}\left[ \Big\| \frac{\lambda^2}{N} \sum_{k=1}^N \rv_k\widetilde\rv_k^\transp - \E\left[ \rv_k\widetilde\rv_k^\transp\right] \Big \|_\infty \ge \sqrt{\lambda^4 (c_1 \frac{\log(M)}{N}+t } )\right ] \\
     &\qquad \leq 2 e^{-c_2 N t}.
	\end{align*}
 
    In particular, the claim holds if $\yv = \widetilde \yv$ and $\yv_i = \widetilde\yv_i, \forall i\in[N]$.
\end{lemma}
\begin{proof}
Write $R_{i,j}^k = [\rv_k]_i  [\widetilde\rv_k]_j$ for $i,j \in [M]$.
	Since $|R_{i,j}^k - \E[R_{i,j}^k]| \le 2, \forall i,j,k$, the bound is trivial for $t\geq 4$. By Bernstein's inequality for bounded random variables (see, e.g., \cite[Theorem 2.8.4]{vershynin2018high}), we find for any $u\leq 8\lambda^2$
	\begin{align*}
	    &\Pr{\frac{1}{N} \left| \sum_{k=1}^N \lambda^2 \round{R_{i,j}^k - \E[R_{i,j}^k] } \right| \ge u} 
	    \le 2e^{-c\min\left\{ \frac{N^2 u^2}{\sigma_{i,j}^2}, \frac{N u}{2\lambda^2}\right\}}\nonumber
	    \\
	    &\qquad \le 2e^{-c N \min\left\{ \frac{u^2}{\lambda^4}, \frac{ u}{ \lambda^2}\right\}}
            \le 2e^{-c_2 N \frac{u^2}{\lambda^4}},
	\end{align*}
	\begin{align*} 
	   \text{as}\quad  \sigma_{i,j}^2 & :=
        \sum_{k=1}^N \lambda^4 \E\left[\left(R_{i,j}^k - \E[R_{i,j}^k]\right)^2\right]\\
       &\;= \sum_{k=1}^N \lambda^4 \left( \E\left[(R_{i,j}^k)^2\right] - \left(\E[R_{i,j}^k]\right)^2 \right)
       \le \lambda^4 N.
    \end{align*}
    Hence, for any given $t<4$ we can set $u=\sqrt{\lambda^4 ( c_1 \frac{\log(M)}{N} + t )}$ and note that $u\leq 8\lambda^2$ as $N\geq c_1\log(M)$. By applying the union bound over all $M^2$ entries we obtain the result. 
\end{proof}

\begin{proof}[Proof of Theorem \ref{thm:FrobeniusDithered}]
By the triangle inequality, we have
\begin{equation}
\label{eqn:TIinfty}
    \pnorm{ \widehat{\Cm}^{\mathrm{d}}_{\yv} -  \Cm_{\yv}}{\infty}\leq \pnorm{ \widehat{\Cm}^{\mathrm{d}}_{\yv} -  \E[\widehat{\Cm}^{\mathrm{d}}_{\yv}]}{\infty} + \pnorm{\E[\widehat{\Cm}^{\mathrm{d}}_{\yv}] - \Cm_{\yv}}{\infty}
\end{equation}
Using notations $\rv, \rv^\Re, \rv^\Im, \widetilde{\rv}, \widetilde{\rv}^\Re, \widetilde{\rv}^\Im$ defined in Lemma \ref{lem:linftyBiasEst}, by Lemma \ref{lem:linftyBiasEst}, we have
\begin{align*}
    \pnorm{\E\left[\widehat{\Cm}^{\mathrm{d}}_{\yv}\right] - \Cm_{\yv}}{\infty} &= \left\| \lambda^2 \E\left[ \rv \widetilde{\rv}^\H \right] - \Cm_{\yv} \right\|_{\infty} \\
    &\lesssim \left(\lambda^2+\|\Cm_{\yv}\|_{\infty}\right)^2 e^{\frac{-c_2 \lambda^2}{\|\Cm_{\yv}\|_{\infty}}} \lesssim \frac{\lambda^2}{\sqrt{N}},\nonumber
\end{align*}
where we have used that $\lambda^2 \gtrsim \log(N) \| \Cm_{\yv} \|_\infty$. To estimate the first term in \eqref{eqn:TIinfty}, observe that
    \begin{align*}
         &\left\| \widehat{\Cm}^{\mathrm{d}}_{\yv} - \E\left[\widehat{\Cm}^{\mathrm{d}}_{\yv}\right] \right\|_\infty
         = \left\| \widetilde{\Cm}^{\mathrm{d}} - \E\left[\widetilde{\Cm}^{\mathrm{d}}\right] \right\|_\infty \\
        & = \Bigg\| \left(\frac{\lambda^2}{N}\sum_{k=1}^N (\rv^\Re_k + \i \rv^\Im_k) (\widetilde{\rv}^\Re_k + \i \widetilde{\rv}^\Im_k )^\H \right) \nonumber\\
        &\quad- \E \left[ \frac{\lambda^2}{N}\sum_{k=1}^N (\rv^\Re_k + \i \rv^\Im_k) (\widetilde{\rv}^\Re_k + \i \widetilde{\rv}^\Im_k )^\H \right] \Bigg\|_\infty  \nonumber\\
        & \le \left\| \left(\frac{\lambda^2}{N}\sum_{k=1}^N \rv^\Re_k (\widetilde{\rv}^\Re_k)^\transp \right) - \E \left[ \frac{\lambda^2}{N}\sum_{k=1}^N \rv^\Re_k (\widetilde{\rv}^\Re_k)^\transp \right] \right\|_\infty\nonumber\\
        &\quad+ \left\| \left(\frac{\lambda^2}{N}\sum_{k=1}^N \rv^\Re_k (\widetilde{\rv}^\Im_k)^\transp \right) - \E \left[ \frac{\lambda^2}{N}\sum_{k=1}^N \rv^\Re_k (\widetilde{\rv}^\Im_k)^\transp \right] \right\|_\infty \nonumber\\
        & \quad+ \left\| \left(\frac{\lambda^2}{N}\sum_{k=1}^N \rv^\Im_k (\widetilde{\rv}^\Re_k)^\transp \right) - \E \left[ \frac{\lambda^2}{N}\sum_{k=1}^N \rv^\Im_k (\widetilde{\rv}^\Re_k)^\transp \right] \right\|_\infty \nonumber\\
        &\quad+ \left\| \left(\frac{\lambda^2}{N}\sum_{k=1}^N \rv^\Im_k (\widetilde{\rv}^\Im_k)^\transp \right) - \E \left[ \frac{\lambda^2}{N}\sum_{k=1}^N \rv^\Im_k (\widetilde{\rv}^\Im_k)^\transp \right] \right\|_\infty \nonumber
    \end{align*}
    Using Lemma \ref{lem:MaxBoundDithered} for each of the four terms and applying a union bound, we get
    \begin{align*} 
        \Pr{  \left\| \widehat{\Cm}^{\mathrm{d}}_{\yv} - \E\left[\widehat{\Cm}^{\mathrm{d}}_{\yv}\right] \right\|_\infty \gtrsim \lambda^2 \sqrt{\frac{\log(M) + t}{N}} } \le 8 e^{-cNt},
    \end{align*}
    and thus the first statement of Theorem~\ref{thm:FrobeniusDithered}. The second statement follows trivially using
    \begin{equation*}
        \pnorm{ \widehat{\Cm}^{\mathrm{d}}_{\yv} -  \Cm_{\yv}}{\sf F} \leq M\pnorm{ \widehat{\Cm}^{\mathrm{d}}_{\yv} -  \Cm_{\yv}}{\infty}.
    \end{equation*}

\end{proof}

\section{Proof of Theorem \ref{thm:MainResult}}
\label{sec:AppendixMainResult}
By Theorem~\ref{thm:FrobeniusDithered}, we can apply Lemma~\ref{lem:Stability} with 
\begin{align*}
    \eps_{\infty} &\sim \lambda^2 \sqrt{\frac{\log(M) + t}{N}}, \\ \eps_{\sf F} &\sim M\max\{1,\|\Cm_{\yv}\|_{\infty}\} \lambda^2 \sqrt{\frac{\log(M) + t}{N}}.
\end{align*}
Note that we do not pick the ``minimal setting'' for $\eps_{\sf F}$ suggested by Theorem~\ref{thm:FrobeniusDithered}: the additional factor $\max\{1,\|\Cm_{\yv}\|_{\infty}\}$ ensures that $\eps_{\infty}\lesssim \eps_{\sf F}/\|\Cm_{\yv}\|_{\sf F}$ holds (as required in \eqref{eqn:epsinftyepsF}). It remains to note that all other conditions on $\eps_{\infty}$ and $\eps_{\sf F}$ in Lemma~\ref{lem:Stability} under the stated assumption on $N$. This completes the proof. 

	{\small
		\bibliographystyle{IEEEtran}
		\bibliography{references}

\begin{thebibliography}{10}
\providecommand{\url}[1]{#1}
\csname url@samestyle\endcsname
\providecommand{\newblock}{\relax}
\providecommand{\bibinfo}[2]{#2}
\providecommand{\BIBentrySTDinterwordspacing}{\spaceskip=0pt\relax}
\providecommand{\BIBentryALTinterwordstretchfactor}{4}
\providecommand{\BIBentryALTinterwordspacing}{\spaceskip=\fontdimen2\font plus
\BIBentryALTinterwordstretchfactor\fontdimen3\font minus
  \fontdimen4\font\relax}
\providecommand{\BIBforeignlanguage}[2]{{%
\expandafter\ifx\csname l@#1\endcsname\relax
\typeout{** WARNING: IEEEtran.bst: No hyphenation pattern has been}%
\typeout{** loaded for the language `#1'. Using the pattern for}%
\typeout{** the default language instead.}%
\else
\language=\csname l@#1\endcsname
\fi
#2}}
\providecommand{\BIBdecl}{\relax}
\BIBdecl

\bibitem{lu2014overview}
L.~Lu, G.~Y. Li, A.~L. Swindlehurst, A.~Ashikhmin, and R.~Zhang, ``An overview
  of massive {MIMO}: {B}enefits and challenges,'' \emph{IEEE Journal of
  Selected Topics in Signal Processing}, vol.~8, no.~5, pp. 742--758, 2014.

\bibitem{bjornson2019massive}
E.~Bjornson, L.~Van~der Perre, S.~Buzzi, and E.~G. Larsson, ``Massive {MIMO} in
  sub-6 {GH}z and mm{W}ave: {P}hysical, practical, and use-case differences,''
  \emph{IEEE Wireless Communications}, vol.~26, no.~2, 2019.

\bibitem{li201712}
Y.~Li, Y.~Luo, G.~Yang \emph{et~al.}, ``12-port {5G} massive {MIMO} antenna
  array in sub-{6GHz} mobile handset for {LTE} bands 42/43/46 applications,''
  \emph{IEEE Access}, vol.~6, pp. 344--354, 2017.

\bibitem{ngo2013energy}
H.~Q. Ngo, E.~G. Larsson, and T.~L. Marzetta, ``Energy and spectral efficiency
  of very large multiuser {MIMO} systems,'' \emph{IEEE Transactions on
  Communications}, vol.~61, no.~4, pp. 1436--1449, 2013.

\bibitem{marzetta2010noncooperative}
T.~L. Marzetta, ``Noncooperative cellular wireless with unlimited numbers of
  base station antennas,'' \emph{IEEE Transactions on Wireless Communications},
  vol.~9, no.~11, pp. 3590--3600, 2010.

\bibitem{haghighatshoar2016massive}
S.~Haghighatshoar and G.~Caire, ``Massive {MIMO} channel subspace estimation
  from low-dimensional projections,'' \emph{IEEE Transactions on Signal
  Processing}, vol.~65, no.~2, pp. 303--318, 2016.

\bibitem{haghighatshoar2017massive}
------, ``Massive {MIMO} pilot decontamination and channel interpolation via
  wideband sparse channel estimation,'' \emph{IEEE Transactions on Wireless
  Communications}, vol.~16, no.~12, pp. 8316--8332, 2017.

\bibitem{haghighatshoar2018low}
------, ``Low-complexity massive {MIMO} subspace estimation and tracking from
  low-dimensional projections,'' \emph{IEEE Transactions on Signal Processing},
  vol.~66, no.~7, pp. 1832--1844, 2018.

\bibitem{khalilsarai2020structured}
M.~B. Khalilsarai, T.~Yang, S.~Haghighatshoar, and G.~Caire, ``Structured
  channel covariance estimation from limited samples in {M}assive {MIMO},'' in
  \emph{IEEE International Conference on Communications (ICC)}, 2020.

\bibitem{payami2012channel}
S.~Payami and F.~Tufvesson, ``Channel measurements and analysis for very large
  array systems at 2.6 ghz,'' in \emph{2012 6th European Conference on Antennas
  and Propagation (EUCAP)}.\hskip 1em plus 0.5em minus 0.4em\relax IEEE, 2012,
  pp. 433--437.

\bibitem{gao2013massive}
X.~Gao, F.~Tufvesson, and O.~Edfors, ``Massive {MIMO} channels—measurements
  and models,'' in \emph{2013 Asilomar Conference on Signals, Systems and
  Computers}.\hskip 1em plus 0.5em minus 0.4em\relax IEEE, 2013, pp. 280--284.

\bibitem{gao2015massive}
X.~Gao, O.~Edfors, F.~Rusek, and F.~Tufvesson, ``Massive {MIMO} performance
  evaluation based on measured propagation data,'' \emph{IEEE Transactions on
  Wireless Communications}, vol.~14, no.~7, pp. 3899--3911, 2015.

\bibitem{wu2014non}
S.~Wu, C.-X. Wang, M.~M. Alwakeel, Y.~He \emph{et~al.}, ``A non-stationary
  3-{D} wideband twin-cluster model for 5{G} massive {MIMO} channels,''
  \emph{IEEE Journal on Selected Areas in Communications}, vol.~32, no.~6,
  2014.

\bibitem{selvan2017fraunhofer}
K.~T. Selvan and R.~Janaswamy, ``Fraunhofer and {F}resnel distances: Unified
  derivation for aperture antennas.'' \emph{IEEE Antennas and Propagation
  Magazine}, vol.~59, no.~4, pp. 12--15, 2017.

\bibitem{bjornson2021primer}
E.~Bj{\"o}rnson, {\"O}.~T. Demir, and L.~Sanguinetti, ``A primer on near-field
  beamforming for arrays and reconfigurable intelligent surfaces,'' in
  \emph{2021 55th Asilomar Conference on Signals, Systems, and
  Computers}.\hskip 1em plus 0.5em minus 0.4em\relax IEEE, 2021, pp. 105--112.

\bibitem{bjornson2020power}
E.~Bj{\"o}rnson and L.~Sanguinetti, ``Power scaling laws and near-field
  behaviors of massive {MIMO} and intelligent reflecting surfaces,'' \emph{IEEE
  Open Journal of the Communications Society}, vol.~1, pp. 1306--1324, 2020.

\bibitem{medbo2014channel}
J.~Medbo, K.~B{\"o}rner, K.~Haneda, V.~Hovinen, T.~Imai, J.~J{\"a}rvelainen,
  T.~J{\"a}ms{\"a}, A.~Karttunen, K.~Kusume, J.~Kyr{\"o}l{\"a}inen
  \emph{et~al.}, ``Channel modelling for the fifth generation mobile
  communications,'' in \emph{The 8th European Conference on Antennas and
  Propagation (EuCAP 2014)}.\hskip 1em plus 0.5em minus 0.4em\relax IEEE, 2014,
  pp. 219--223.

\bibitem{bjornson2019massivenext}
E.~Bj{\"o}rnson, L.~Sanguinetti, H.~Wymeersch, J.~Hoydis, and T.~L. Marzetta,
  ``Massive {MIMO} is a reality—what is next?: Five promising research
  directions for antenna arrays,'' \emph{Digital Signal Processing}, vol.~94,
  pp. 3--20, 2019.

\bibitem{de2020non}
E.~De~Carvalho, A.~Ali, A.~Amiri, M.~Angjelichinoski, and R.~W. Heath,
  ``Non-stationarities in extra-large-scale massive {MIMO},'' \emph{IEEE
  Wireless Communications}, vol.~27, no.~4, pp. 74--80, 2020.

\bibitem{walden1999analog}
R.~H. Walden, ``Analog-to-digital converter survey and analysis,'' \emph{IEEE
  Journal on Selected Areas in Communications}, vol.~17, no.~4, pp. 539--550,
  1999.

\bibitem{murmann2015race}
B.~Murmann, ``The race for the extra decibel: A brief review of current {ADC}
  performance trajectories,'' \emph{IEEE Solid-State Circuits Magazine},
  vol.~7, no.~3, pp. 58--66, 2015.

\bibitem{mezghani2008analysis}
A.~Mezghani and J.~A. Nossek, ``Analysis of {R}ayleigh-fading channels with
  1-bit quantized output,'' in \emph{2008 IEEE International Symposium on
  Information Theory}.\hskip 1em plus 0.5em minus 0.4em\relax IEEE, 2008, pp.
  260--264.

\bibitem{nossek2006capacity}
J.~A. Nossek and M.~T. Ivrla{\v{c}}, ``Capacity and coding for quantized {MIMO}
  systems,'' in \emph{Proceedings of the 2006 International Conference on
  Wireless Communications and Mobile Computing}, 2006, pp. 1387--1392.

\bibitem{singh2009limits}
J.~Singh, O.~Dabeer, and U.~Madhow, ``On the limits of communication with
  low-precision analog-to-digital conversion at the receiver,'' \emph{IEEE
  Transactions on Communications}, vol.~57, no.~12, pp. 3629--3639, 2009.

\bibitem{cheng2019adaptive}
X.~Cheng, K.~Xu, J.~Sun, and S.~Li, ``Adaptive grouping sparse {B}ayesian
  learning for channel estimation in non-stationary uplink massive {MIMO}
  systems,'' \emph{IEEE Transactions on Wireless Communications}, vol.~18,
  no.~8, pp. 4184--4198, 2019.

\bibitem{han2020deep}
Y.~Han, M.~Li, S.~Jin, C.-K. Wen, and X.~Ma, ``Deep learning-based {FDD}
  non-stationary massive {MIMO} downlink channel reconstruction,'' \emph{IEEE
  Journal on Selected Areas in Communications}, vol.~38, no.~9, 2020.

\bibitem{li2017channel}
Y.~Li, C.~Tao, G.~Seco-Granados, A.~Mezghani, A.~L. Swindlehurst, and L.~Liu,
  ``Channel estimation and performance analysis of one-bit massive {MIMO}
  systems,'' \emph{IEEE Transactions on Signal Processing}, vol.~65, no.~15,
  pp. 4075--4089, 2017.

\bibitem{wan2020generalized}
Q.~Wan, J.~Fang, H.~Duan, Z.~Chen, and H.~Li, ``Generalized {B}ussgang {LMMSE}
  channel estimation for one-bit massive {MIMO} systems,'' \emph{IEEE
  Transactions on Wireless Communications}, vol.~19, no.~6, 2020.

\bibitem{eamaz2021modified}
A.~Eamaz, F.~Yeganegi, and M.~Soltanalian, ``Modified arcsine law for one-bit
  sampled stationary signals with time-varying thresholds,'' in \emph{ICASSP
  2021-2021 IEEE International Conference on Acoustics, Speech and Signal
  Processing (ICASSP)}.\hskip 1em plus 0.5em minus 0.4em\relax IEEE, 2021, pp.
  5459--5463.

\bibitem{eamaz2022covariance}
------, ``Covariance recovery for one-bit sampled non-stationary signals with
  time-varying sampling thresholds,'' \emph{IEEE Transactions on Signal
  Processing}, vol.~70, pp. 5222--5236, 2022.

\bibitem{eamaz2023covariance}
------, ``Covariance recovery for one-bit sampled stationary signals with
  time-varying sampling thresholds,'' \emph{Signal Processing}, vol. 206, p.
  108899, 2023.

\bibitem{chapeau2008fisher}
F.~Chapeau-Blondeau, S.~Blanchard, and D.~Rousseau, ``Fisher information and
  noise-aided power estimation from one-bit quantizers,'' \emph{Digital Signal
  Processing}, vol.~18, no.~3, pp. 434--443, 2008.

\bibitem{fang2010adaptive}
J.~Fang and H.~Li, ``Adaptive distributed estimation of signal power from
  one-bit quantized data,'' \emph{IEEE Transactions on Aerospace and Electronic
  Systems}, vol.~46, no.~4, pp. 1893--1905, 2010.

\bibitem{liu2021one}
C.-L. Liu and Z.-M. Lin, ``One-bit autocorrelation estimation with non-zero
  thresholds,'' in \emph{ICASSP 2021-2021 IEEE International Conference on
  Acoustics, Speech and Signal Processing (ICASSP)}.\hskip 1em plus 0.5em minus
  0.4em\relax IEEE, 2021, pp. 4520--4524.

\bibitem{xiao2023one}
Y.-H. Xiao, L.~Huang, D.~Ram{\'\i}rez, C.~Qian, and H.~C. So, ``One-bit
  covariance reconstruction with non-zero thresholds: Algorithm and performance
  analysis,'' \emph{arXiv preprint arXiv:2303.16455}, 2023.

\bibitem{dirksen2021covariance}
S.~Dirksen, J.~Maly, and H.~Rauhut, ``{Covariance estimation under one-bit
  quantization},'' \emph{The Annals of Statistics}, vol.~50, no.~6, pp. 3538 --
  3562, 2022.

\bibitem{bill2022nnls}
\BIBentryALTinterwordspacing
B.~Whiten. (2013) {NNLS} - non negative least squares. [Online]. Available:
  \url{https://www.mathworks.com/matlabcentral/fileexchange/38003-nnls-non-negative-least-squares}
\BIBentrySTDinterwordspacing

\bibitem{3gpp38901}
3GPP, ``Study on channel model for frequencies from 0.5 to 100 ghz (release
  15),'' 3rd Generation Partnership Project (3GPP), Tech. Rep. TR 38.901
  V15.0.0, 2018.

\bibitem{3gpp25996}
------, ``Spatial channel model for multiple input multiple output ({MIMO})
  simulations (release 16),'' 3rd Generation Partnership Project (3GPP), Tech.
  Rep. TR 25.996 V16.0.0, 2020.

\bibitem{jaeckel2014quadriga}
S.~Jaeckel, L.~Raschkowski, K.~B{\"o}rner, and L.~Thiele, ``{QuaDRiGa}: A 3-{D}
  multi-cell channel model with time evolution for enabling virtual field
  trials,'' \emph{IEEE Transactions on Antennas and Propagation}, vol.~62,
  no.~6, pp. 3242--3256, 2014.

\bibitem{va2016impact}
V.~Va, J.~Choi, and R.~W. Heath, ``The impact of beamwidth on temporal channel
  variation in vehicular channels and its implications,'' \emph{IEEE
  Transactions on Vehicular Technology}, vol.~66, no.~6, 2016.

\bibitem{bar2002doa}
O.~Bar-Shalom and A.~J. Weiss, ``{DOA} estimation using one-bit quantized
  measurements,'' \emph{IEEE Transactions on Aerospace and Electronic Systems},
  vol.~38, no.~3, pp. 868--884, 2002.

\bibitem{risi2014massive}
C.~Risi, D.~Persson, and E.~G. Larsson, ``{Massive MIMO with 1-bit ADC},''
  \emph{arXiv preprint arXiv:1404.7736}, 2014.

\bibitem{bussgang1952crosscorrelation}
J.~J. Bussgang, ``Crosscorrelation functions of amplitude-distorted {G}aussian
  signals,'' \emph{Research Laboratory of Electronics, Massachusetts Institute
  of Technology}, 1952.

\bibitem{demir2020bussgang}
O.~T. Demir and E.~Bjornson, ``The {B}ussgang decomposition of nonlinear
  systems: Basic theory and {MIMO} extensions [lecture notes],'' \emph{IEEE
  Signal Processing Magazine}, vol.~38, no.~1, pp. 131--136, 2020.

\bibitem{mezghani2012capacity}
A.~Mezghani and J.~A. Nossek, ``Capacity lower bound of {MIMO} channels with
  output quantization and correlated noise,'' in \emph{Proc. IEEE Int. Symp.
  Inf. Theory}, 2012, pp. 1--5.

\bibitem{papoulis2002probability}
A.~Papoulis and S.~U. Pillai, \emph{Probability, Random Variables, and
  Stochastic Processes}.\hskip 1em plus 0.5em minus 0.4em\relax Tata
  McGraw-Hill Education, 2002.

\bibitem{van1966spectrum}
J.~H. Van~Vleck and D.~Middleton, ``The spectrum of clipped noise,''
  \emph{Proceedings of the IEEE}, vol.~54, no.~1, pp. 2--19, 1966.

\bibitem{jacovitti1994estimation}
G.~Jacovitti and A.~Neri, ``Estimation of the autocorrelation function of
  complex {G}aussian stationary processes by amplitude clipped signals,''
  \emph{IEEE Transactions on Information Theory}, vol.~40, no.~1, pp. 239--245,
  1994.

\bibitem{GrN98}
R.~M. Gray and D.~L. Neuhoff, ``Quantization,'' \emph{IEEE Transactions on
  Information Theory}, vol.~44, no.~6, pp. 2325--2383, 1998.

\bibitem{GrS93}
R.~M. Gray and T.~G. Stockham, ``Dithered quantizers,'' \emph{IEEE Transactions
  on Information Theory}, vol.~39, no.~3, pp. 805--812, 1993.

\bibitem{Rob62}
L.~Roberts, ``Picture coding using pseudo-random noise,'' \emph{IRE
  Transactions on Information Theory}, vol.~8, no.~2, pp. 145--154, 1962.

\bibitem{BFN17}
R.~G. Baraniuk, S.~Foucart, D.~Needell, Y.~Plan, and M.~Wootters, ``Exponential
  decay of reconstruction error from binary measurements of sparse signals,''
  \emph{IEEE Transactions on Information Theory}, vol.~63, no.~6, pp.
  3368--3385, 2017.

\bibitem{Dir19}
S.~Dirksen, ``Quantized compressed sensing: {A} {S}urvey,'' in \emph{Compressed
  Sensing and Its Applications}.\hskip 1em plus 0.5em minus 0.4em\relax
  Springer, 2019, pp. 67--95.

\bibitem{dirksen2023robust}
S.~Dirksen and S.~Mendelson, ``Robust one-bit compressed sensing with partial
  circulant matrices,'' \emph{The Annals of Applied Probability}, vol.~33,
  no.~3, pp. 1874--1903, 2023.

\bibitem{DiM18a}
------, ``Non-{G}aussian hyperplane tessellations and robust one-bit compressed
  sensing,'' \emph{Journal of the European Mathematical Society}, vol.~23,
  no.~9, pp. 2913--2947, 2021.

\bibitem{JMP19}
H.~C. Jung, J.~Maly, L.~Palzer, and A.~Stollenwerk, ``Quantized compressed
  sensing by rectified linear units,'' \emph{IEEE Transactions on Information
  Theory}, vol.~67, no.~6, pp. 4125--4149, 2021.

\bibitem{KSW16}
K.~Knudson, R.~Saab, and R.~Ward, ``One-bit compressive sensing with norm
  estimation,'' \emph{IEEE Transactions on Information Theory}, vol.~62, no.~5,
  pp. 2748--2758, 2016.

\bibitem{eamaz2022uno}
A.~Eamaz, K.~V. Mishra, F.~Yeganegi, and M.~Soltanalian, ``{UNO: Unlimited
  sampling meets one-bit quantization},'' \emph{arXiv preprint
  arXiv:2301.10155}, 2022.

\bibitem{khalilsarai2018fdd}
M.~B. Khalilsarai, S.~Haghighatshoar, X.~Yi, and G.~Caire, ``{FDD} massive
  {MIMO} via {UL}/{DL} channel covariance extrapolation and active channel
  sparsification,'' \emph{IEEE Transactions on Wireless Communications},
  vol.~18, no.~1, pp. 121--135, 2018.

\bibitem{khalilsarai2021dual}
M.~B. Khalilsarai, T.~Yang, S.~Haghighatshoar, X.~Yi, and G.~Caire,
  ``Dual-polarized {FDD} massive {MIMO}: A comprehensive framework,''
  \emph{IEEE Transactions on Wireless Communications}, vol.~21, no.~2, 2021.

\bibitem{yang2023structured}
T.~Yang, M.~Barzegar~Khalilsarai, S.~Haghighatshoar, and G.~Caire, ``Structured
  channel covariance estimation from limited samples for large antenna
  arrays,'' \emph{EURASIP Journal on Wireless Communications and Networking},
  vol. 2023, no.~1, pp. 1--28, 2023.

\bibitem{chen2010nonnegativit}
D.~Chen and R.~J. Plemmons, ``Nonnegativity constraints in numerical
  analysis,'' in \emph{The birth of numerical analysis}.\hskip 1em plus 0.5em
  minus 0.4em\relax World Scientific, 2010, pp. 109--139.

\bibitem{lawson1974solving}
C.~L. Lawson and R.~Hanson, \emph{Solving least squares problems}.\hskip 1em
  plus 0.5em minus 0.4em\relax Prentice-Hall, Englewood Cliffs, NJ, 1974.

\bibitem{caire2018ergodic}
G.~Caire, ``{On the ergodic rate lower bounds with applications to massive
  MIMO},'' \emph{IEEE Transactions on Wireless Communications}, vol.~17, no.~5,
  pp. 3258--3268, 2018.

\bibitem{diggavi2001worst}
S.~N. Diggavi and T.~M. Cover, ``The worst additive noise under a covariance
  constraint,'' \emph{IEEE Transactions on Information Theory}, vol.~47, no.~7,
  pp. 3072--3081, 2001.

\bibitem{max1960quantizing}
J.~Max, ``Quantizing for minimum distortion,'' \emph{IRE Transactions on
  Information Theory}, vol.~6, no.~1, pp. 7--12, 1960.

\bibitem{desset2015validation}
C.~Desset and L.~Van~der Perre, ``{Validation of low-accuracy quantization in
  massive MIMO and constellation EVM analysis},'' in \emph{2015 European
  Conference on Networks and Communications (EuCNC)}.\hskip 1em plus 0.5em
  minus 0.4em\relax IEEE, 2015, pp. 21--25.

\bibitem{vershynin2018high}
R.~Vershynin, \emph{High-dimensional probability: An introduction with
  applications in data science}.\hskip 1em plus 0.5em minus 0.4em\relax
  Cambridge University Press, 2018, vol.~47.

\end{thebibliography}
	}
	
\end{document}